\documentclass[12pt,letterpaper]{article}

\usepackage{natbib}
\usepackage[ left=1in, top=1in, right=1in, bottom=1in]{geometry}
\usepackage{xcolor,graphicx,bm,colonequals,amsmath,amssymb,url}
\usepackage{array,tabularx,multirow}
\usepackage[shortlabels]{enumitem}
\usepackage[font={footnotesize}]{caption,subcaption}
\usepackage{color}

\usepackage[normalem]{ulem} 
\usepackage{bbm}

\definecolor{joe}{rgb}{0.5,0,1}

\bibpunct[, ]{(}{)}{;}{a}{,}{,}

\usepackage{amsthm}
\newtheoremstyle{propstyle} 
    {3mm}                    
    {1mm}                    
    {\itshape}                   
    {}                           
    {\scshape}                   
    {.}                          
    {.5em}                       
    {}  
\theoremstyle{propstyle}
\newtheorem{prop}{Proposition}
\theoremstyle{propstyle}

\theoremstyle{propstyle}

\newsavebox\ideabox

\newcommand{\ba}{\mathbf{a}}

\newcommand{\bs}{\mathbf{s}}

\newcommand{\bx}{\mathbf{x}}
\newcommand{\by}{\mathbf{y}}

\newcommand{\bz}{\mathbf{z}}

\newcommand{\bP}{\mathbf{P}}
\newcommand{\bA}{\mathbf{A}}

\newcommand{\bW}{\mathbf{W}}

\newcommand{\bI}{\mathbf{I}}
\newcommand{\bD}{\mathbf{D}}

\newcommand{\bU}{\mathbf{U}}
\newcommand{\bV}{\mathbf{V}}

\newcommand{\bB}{\mathbf{B}}
\newcommand{\bC}{\mathbf{C}}
\newcommand{\bM}{\mathbf{M}}

\newcommand{\beps}{\bm{\varepsilon}}

\newcommand{\all}{\bullet}

\newcommand{\bfzero}{\mathbf{0}}

\newcommand{\bfmu}{\bm{\mu}}

\newcommand{\bftheta}{\bm{\theta}}

\DeclareMathOperator{\E}{E}
\DeclareMathOperator{\var}{var}
\DeclareMathOperator{\cov}{cov}

\DeclareMathOperator{\chol}{chol}
\DeclareMathOperator{\rchol}{rchol}

\newcommand{\GP}{GP}

\renewcommand{\path}{Q}
\newcommand{\normal}{\mathcal{N}}
\newcommand{\order}{\mathcal{O}}

\newcommand{\domain}{\mathcal{D}}

\newcommand{\dens}{f}
\newcommand{\adens}{\widehat{f}}

\newcommand{\locs}{\mathcal{S}}

\title{A general framework for Vecchia approximations of Gaussian processes}

\author{
Matthias Katzfuss\thanks{Department of Statistics, Texas A\&M University. \texttt{katzfuss@gmail.com}}
\and
Joseph Guinness\thanks{Department of Statistics and Data Science, Cornell University}
}

\date{}

\begin{document}

\maketitle

\begin{abstract}
Gaussian processes (GPs) are commonly used as models for functions, time series, and spatial fields, but they are computationally infeasible for large datasets. Focusing on the typical setting of modeling data as a GP plus an additive noise term, we propose a generalization of the Vecchia (1988) approach as a framework for GP approximations. We show that our general Vecchia approach contains many popular existing GP approximations as special cases, allowing for comparisons among the different methods within a unified framework. Representing the models by directed acyclic graphs, we determine the sparsity of the matrices necessary for inference, which leads to new insights regarding the computational properties. Based on these results, we propose a novel sparse general Vecchia approximation, which ensures computational feasibility for large spatial datasets but can lead to considerable improvements in approximation accuracy over Vecchia's original approach. We provide several theoretical results and conduct numerical comparisons. We conclude with guidelines for the use of Vecchia approximations in spatial statistics.
\end{abstract}


\section{Introduction \label{sec:intro}}

Gaussian processes (GPs) have become popular choices as models or prior distributions for functions, time series, and spatial fields \citep[e.g.,][]{Banerjee2004,Rasmussen2006,Cressie2011}. The defining feature of a GP is that the joint distribution of a finite number of observations is multivariate normal. However, since computing with multivariate normal distributions incurs quadratic memory and cubic time complexity in the number of observations, GP inference is infeasible when the data size is in the tens of thousands or higher, limiting the direct use of GPs for many large datasets available today.

To achieve computational feasibility, numerous approaches have been proposed in the statistics and machine-learning literatures.
These include approaches leading to sparse covariance matrices \citep{furrer2006covariance,kaufman2008covariance,Du2009}, sparse inverse covariance (i.e., precision) matrices \citep{rue2005gaussian,lindgren2011explicit,Nychka2012}, and low-rank matrices \citep[e.g.,][]{Higdon1998, Wikle1999, Quinonero-Candela2005, Banerjee2008, Cressie2008, Katzfuss2010}. Several other approaches are described in Section \ref{sec:existing}.
\citet{Heaton2017} review and compare many of these methods, plus several algorithmic approaches \citep{Gramacy2015,Gerber2018,Guhaniyogi2018}.

In this article, we extend and study Vecchia's approach \citep{Vecchia1988}, one of the earliest proposed GP approximations, which leads to a sparse Cholesky factor of the precision matrix. Based on some ordering of the GP observations, Vecchia's approximation replaces the high-dimensional joint distribution with a product of univariate conditional distributions, in which each conditional distribution conditions on only a small subset of previous observations in the ordering. This approximation incurs low computational and memory burden, it has been shown to be highly accurate in terms of Kullback-Leibler divergence from the true model \citep[e.g.,][]{Guinness2016a}, and it is amenable to parallel computing because each term can be computed separately.

We consider the typical setting of spatial data modeled as a GP plus an additive noise or nugget component. \citet{Datta2016} proposed to apply Vecchia's approximation to the latent GP instead of the noisy observations, but \citet{Finley2017} noted that this approach ``require[d] an excessively long run time.'' Here, we propose a generalized version of the Vecchia approximation, which allows conditioning on both latent and observed variables. We show that our general Vecchia approach contains several popular GP approximations as special cases, allowing for comparisons among the different approaches within a unified framework. We give a formula for efficient computation of the likelihood in the presence of noise. Further, we describe how approximations within the general Vecchia framework can be represented by directed acyclic graph (DAG) models, and we use the connection to DAGs to prove results about the sparsity of the matrices appearing in the inference algorithms. The results lead to new insights regarding computational properties, including shedding light on the computational challenges with latent Vecchia noted in \citet{Finley2017}. Based on these results, we propose a particular instance of the general Vecchia framework, which we call sparse general Vecchia (SGV), that provides guaranteed levels of sparsity in its matrix representation but can lead to considerable improvements in approximation accuracy over Vecchia's original approach. In addition to the theoretical results, we provide numerical studies exploring different options within the general Vecchia framework and comparing our novel SGV to existing approximations.

This article is organized as follows. In Section \ref{sec:methodology}, we review Vecchia's approximation, introduce our general Vecchia framework, and detail connections to DAGs. In Section \ref{sec:existing}, we describe several existing GP approximations as special cases of the framework. In Section \ref{sec:computation}, we consider inference within the framework, including introducing the necessary matrices and studying their sparsity, and deriving the computational complexity. In Section \ref{sec:sgv}, we describe our new SGV approximation and contrast it with two existing approaches.
Section \ref{sec:ordcond} contains additional insights on ordering and conditioning.
Numerical results and comparisons can be found in Section \ref{sec:numerical}.
In Section \ref{sec:conclusions}, we conclude and provide guidelines for the use of Vecchia approximations.
Appendices \ref{app:vectornotation}--\ref{app:proofs} contain further details and proofs. The methods and algorithms proposed here are implemented in the R package \texttt{GPvecchia} available at \url{https://github.com/katzfuss-group/GPvecchia}.


\section{A general Vecchia approach \label{sec:methodology}}

\subsection{Noisy observations of a Gaussian process \label{sec:gpintro}}

Let $\{y(\bs) \!: \bs \in \domain\}$, or $y(\cdot)$, be a process of interest on a continuous (i.e., non-gridded) domain $\domain \subset \mathbb{R}^d$, $d \in \mathbb{N}^+$. We assume that $y(\cdot) \sim \GP(0,K)$ is a zero-mean Gaussian process (GP) with covariance function $K: \domain \times \domain \to \mathbb{R}$. We place no restrictions on $K$, other than assuming that it is a positive-definite function that is known up to a vector of parameters, $\bftheta$. Usually, $K$ will be a continuous covariance function without a nugget component, which will be added in the next paragraph. In most applications, $y(\cdot)$ will not have zero mean, but estimating and subtracting the mean is typically not a computational problem, so we ignore the mean here for simplicity.
Further, let $\locs$ be a vector of vectors of locations, meaning that $\locs = (\locs_1,\ldots,\locs_\ell)$, where $\locs_i$ is a vector of $r_i$ locations in $\domain$. (Our vector and indexing notation is explained in detail in Appendix \ref{app:vectornotation}.) Then define $\by_i = y(\locs_i)$ to be the Gaussian process vector at locations $\locs_i$, and form the vector $\by\colonequals (\by_1,\ldots,\by_\ell)$.

We observe $\bz_i = \by_i + \beps_i$, where the noise or nugget terms $\beps_i$ are independent $\normal_{r_i}(\bfzero,\tau^2 \bI)$. The noisy-observation assumption is ubiquitous in spatial statistics, GP regression, and functional data, and has been proposed for the modeling of computer experiments \citep{gramacy2012cases}. In this work, we assume that we observe the subset $\bz_o$ of $\bz = (\bz_1,\ldots,\bz_\ell)$, where $o \subset (1,\ldots,\ell)$.
Parameters $\bftheta$ and $\tau^2$ are assumed to be known for now; parameter inference will be discussed in Section \ref{sec:likelihood}.

\subsection{Review of Vecchia's approximation \label{sec:vecchiareview}}

Define $h_o(i) \colonequals o \cap (1,\ldots,i-1)$ to be the observed ``history'' of $i$ with $h_o(1) = \emptyset$, allowing us to write the joint density for the observed vector $\bz_o$ as
\begin{equation}
\label{eq:exactdecomp}
\dens(\bz_o) = \prod_{i \in o} \dens(\bz_{i} | \bz_{h_o(i)} ).
\end{equation}
Working with or evaluating the density in \eqref{eq:exactdecomp} directly incurs $O(n_z^2)$ memory and $\order(n_z^3)$ computational cost, and is thus infeasible for large $n_z$, where $n_z$ is the number of individual observations in $\bz_o$.


To avoid these computational difficulties, Vecchia's approximation \citep{Vecchia1988} replaces $h_o(i)$ with a subvector $g(i)$, where $g(i)$ is often chosen to contain those indices corresponding to observations nearby in distance to the $i$th vector of observations. We refer to $g(i)$ as the $i$th conditioning index vector, and to $\bz_{g(i)}$ as the conditioning vector for $\bz_i$. This leads to the Vecchia approximation of the joint density in \eqref{eq:exactdecomp}:
\begin{equation}
\label{eq:approxdecomp}
\adens(\bz_o) = \prod_{i \in o} \dens(\bz_i | \bz_{g(i)}).
\end{equation}
\cite{Vecchia1988} considered only the case of $\bz_i$ as singletons, whereas \cite{stein2004} are credited with the generalization to vector $\bz_i$. \citet{cressie1998image} showed that \eqref{eq:approxdecomp} implies a Markov random field model with sparse precision matrix. \cite{stein2004} showed that maximizing \eqref{eq:approxdecomp} corresponds to solving a set of unbiased estimating equations, and they proposed a residual maximum likelihood (REML) method for estimating covariance parameters.

\subsection{The general Vecchia framework \label{sec:generalvecchia}}

The standard Vecchia approach in Section \ref{sec:vecchiareview} applies to the vector of observations, $\bz_o$.
We propose a general Vecchia approach, which applies Vecchia's approximation to a vector $\bx = \by \cup \bz_o$ consisting of the data $\bz_o$ and latent variables $\by$:
\begin{equation}
\label{eq:vecchia1}
\adens(\bx) = \prod_{i=1}^b f( \bx_i | \bx_{g(i)} ),
\end{equation}
where $b$ is the number of subvectors in $\bx$, and $g(i) \subset h(i) = (1,\ldots,i-1)$.
Here, the elements of $\by$ and $\bz_o$ are interweaved within $\bx$. Specifically, using the notation from Appendix \ref{app:vectornotation}, the ordering in $\bx$ is defined as $\#(\by_i,\bx) < \#(\by_j,\bx)$ when $i<j$, and $\#(\bz_i,\bx) = \#(\by_i,\bx)+1$. In words, the $\by_i$ vectors retain their relative ordering in $\bx$, and $\bz_i$ is inserted directly after $\by_i$ when $i \in o$.
Then, the general Vecchia approximation in \eqref{eq:vecchia1} can be written as
\begin{equation}
\label{eq:vecchia2}
\adens(\bx) = \bigg(\prod_{i=1}^\ell f ( \, \by_i \, | \, \by_{q_y(i)},  \bz_{q_z(i)}  \, ) \bigg) \bigg(\prod_{i \in o} f( \bz_i | \by_i ) \bigg).
\end{equation}
For the conditioning vector of $\by_i$, $j \in q_y(i)$ means that $\by_i$ conditions on $\by_j$, while $j \in q_z(i)$ means that $\by_i$ conditions on $\bz_j$. It can be more accurate but also more computationally expensive to condition on $\by_j$ rather than on $\bz_j$; we will explore this tradeoff in Section \ref{sec:sgv}. We always pick $\by_i$ as the conditioning vector for $\bz_i$, because $\bz_i$ was defined to be conditionally independent of all other vectors given $\by_i$. For the same reason, there is nothing to be gained by conditioning $\by_i$ on both $\by_j$ and $\bz_j$, and so we always take $q_y(i) \cap q_z(i) = \emptyset$. We call $q(i) = \big( q_y(i) , q_z(i) \big)$ the conditioning index vector.
Note that if $j\notin o$, assuming $j \in q_z(i)$ is equivalent to removing $j$ from $q(i)$.

Usually, it is of interest to evaluate an approximation to $f(\bz_o)$, which involves integrating the approximation for the joint density of $\bz_o$ and $\by$ in \eqref{eq:vecchia1} and \eqref{eq:vecchia2} over the latent vector $\by$:
\begin{equation}
\label{eq:induceddist}
\textstyle\adens(\bz_o) = \int \adens(\bx) d\by.
\end{equation}
If the conditioning vectors are equal to the respective history vectors (i.e., $g(i) = h(i)$ for all $i$), the exact distribution $\dens(\bz_o)$ in \eqref{eq:exactdecomp} is recovered. In this sense, the general Vecchia approximation converges to the truth as the conditioning vectors grow larger. However, large conditioning vectors negate the computational advantages, and thus the case of small conditioning vectors is of interest here.

In summary, a general Vecchia approximation $\adens(\bx)$ of $\dens(\bx)$, and the implied approximation $\adens(\bz_o)$ of $\dens(\bz_o)$, are determined by the following choices:
\begin{description}[itemsep=0.3pt,topsep=10pt,parsep=3pt]
\item[    C1:] The $n_y$ locations $\locs$, usually a superset of the observed locations.
\item[    C2:] The partitioning of $\locs$ into $\ell \leq n_y$ vectors of locations.
\item[    C3:] The ordering of the location vectors as $\locs = (\locs_1,\ldots,\locs_\ell)$.
\item[    C4:] For each $i$, the conditioning index vector $q(i) \subset (1,\ldots,i-1)$ for $\by_i$.
\item[    C5:] For each $i$, the partitioning of $q(i)$ into $q_y(i)$ and $q_z(i)$; that is, for each $j \in q(i)$, whether $\by_i$ should condition on $\by_j$ or $\bz_j$.
\end{description}
For C1, the default choice is often to set $\locs$ equal to the observed locations. A major focus of this article is C5, which is discussed in Section \ref{sec:sgv}. In Section \ref{sec:ordcond}, we provide some insights into C2--C4, and in Section \ref{sec:numerical} we explore C3--C5 numerically.

\subsection{Connections to directed acyclic graphs \label{sec:dags}}

There are strong connections between the Vecchia approach and directed acyclic graphs \citep[DAGs; cf.][]{Datta2016}. A brief review of DAGs is provided in Appendix \ref{app:dags}. The conditional-independence structure implied by the Vecchia approximation in \eqref{eq:vecchia1} can be well represented by a DAG. Viewing $\bx_1,\ldots,\bx_b$ as the vertices in the DAG, we have $\bx_j \to \bx_i$ if and only if $j \in g(i)$, and so $\bx_{g(i)}$ is the vector formed by the set of all parents of $\bx_i$. Note that, because Vecchia approximations allow conditioning only on previous variables in the ordering, we always have $\bx_i \not\to \bx_j$ if $i>j$.
DAG representations are illustrated in Figure \ref{fig:toyillustration}.
We will use this connection between Vecchia approaches and DAGs to study the sparsity of the matrices needed for inference in Section \ref{sec:sparsity}.


\section{Existing methods as special cases \label{sec:existing}}

Many existing GP approximations fall into the framework described above. Each of these special cases corresponds to particular choices of C1--C5. We give some examples here. Most of these examples are illustrated in Figure \ref{fig:toyillustration}.

\subsection{Standard Vecchia and extensions \label{sec:vecchia}}

Vecchia's original approximation \citep{Vecchia1988} specifies singleton vectors ($r_i=1$), ordering locations by a spatial coordinate (henceforth referred to as coord ordering), and conditioning only on observations $\bz_i$ (as opposed to latent $\by_i$); that is, $q_z(i) = q(i)$, $q_y(i) = \emptyset$, and $o=(1,\ldots,\ell)$. Using \eqref{eq:induceddist}, this results in the approximation
\begin{align*}
\textstyle \adens(\bz) = \int \prod_{i=1}^\ell \dens(\bz_i|\by_i) \dens(\by_i | \bz_{q(i)})  d\by = \prod_{i=1}^\ell \int \dens(\bz_i|\by_i,\bz_{q(i)}) \dens(\by_i | \bz_{q(i)} ) d\by_i = \prod_{i=1}^\ell \dens(\bz_i|\bz_{q(i)}),
\end{align*}
where $\dens(\bz_i|\by_i) = \dens(\bz_i|\by_i,\bz_{q(i)})$ because $\bz_i$ is conditionally independent of $\bz_{q(i)}$ given $\by_i$.
\citet{stein2004} recommended including in $q(i)$ the indices of some close and some far-away observations, and grouping observations (i.e., $r_i>1$) for computational advantages. \citet{Guinness2016a} considered an adaptive grouping scheme, and discovered that ordering schemes other than coord ordering can improve approximation accuracy. \citet{Vecchia1988} and \citet{stein2004} focus on likelihood approximation, but \citet{Guinness2016a} also considers spatial prediction via conditional simulation. Further extensions were proposed in \citet{Sun2016} and \citet{Huang2016}. Some asymptotics are provided in \citet{Zhang2012}.

\subsection{Nearest-neighbor GP (NNGP) \label{sec:nngp}}

The NNGP \citep{Datta2016,Datta2016b,Datta2016a} considers explicit data models (such as the additive Gaussian noise assumed here), and conditions only on latent variables: $q_y(i) = q(i)$ and $q_z(i) = \emptyset$. A GP is defined by setting $r_i = 1$, $\locs = (\locs_1,\ldots,\locs_{\ell_y+\ell_z})$, $o = (\ell_y+1,\ldots,\ell_y+\ell_z)$, and enforcing the constraint that $q(i) \subset (1,\ldots,\ell_y)$ for all $i$. This means that variables at the observed locations can condition only on variables in the knot set $(\locs_1,\ldots,\locs_{\ell_y})$.

\subsection{Independent blocks \label{sec:indblocks}}

The simplest special case is given by empty conditioning index vectors $q(i) = \emptyset$ for every $i$:
\begin{align*}
\textstyle \adens(\bz) = \int \prod_{i=1}^\ell \dens(\bz_i|\by_i) \dens(\by_i )  d\by =  \prod_{i=1}^\ell \int \dens(\bz_i|\by_i) \dens(\by_i )  d\by = \prod_{i=1}^\ell \dens(\bz_i),
\end{align*}
which treats the $\ell$ subvectors $\bz_1,\ldots,\bz_\ell$ independently and assumes $o=(1,\ldots,\ell)$. When each subvector or block corresponds to a contiguous subregion in space, \citet{Stein2013a} showed that this approximation can be quite competitive as a surrogate for the likelihood, and also computationally inexpensive since each term in the product incurs small computational cost and can be computed in parallel. One difficulty with this approach is characterization of joint uncertainties in predictions, due to the independence assumption embedded in the approximation.

\subsection{Latent autoregressive process of order $m$ (AR(m))}

Latent (vector-)AR processes of order $m$, also called state-space models, are common in time-series settings. They condition only on the latest $m$ sets of latent variables for some ordering: $q(i) = q_y(i) = (i-m,\ldots,i-1)$. Inference in this type of model is typically carried out using the Kalman filter and smoother \citep{Kalman1960,Rauch1965}.

\subsection{Modified predictive process (MPP) \label{sec:mpp}}

The MPP \citep[][]{Finley2009} is obtained by defining $\locs_1$ as a vector of ``knot'' locations, typically $1 \notin o$, and for all $i>1$, set $r_i = 1$ and $q(i) = q_y(i) = 1$. This means that all variables condition on the same vector $\by_1$ and are assumed to be conditionally independent.

\subsection{Full-scale approximation (FSA) \label{sec:fsa}}

As in the MPP, the FSA-block \citep{Snelson2007,Sang2011a} is obtained by designating a common conditioning vector $\by_1$, and setting $q(i) = q_y(i) = 1$ for $i>1$. However, the FSA allows $r_i > 1$ and groups all remaining variables by spatial region as in the independent-blocks case in Section \ref{sec:indblocks}.

When $q(i) = q_y(i) = 1$ for all $i>1$ in general Vecchia, we have $\adens(\by) = \dens(\by_1) \prod_{i=2}^\ell \dens(\by_i | \by_1)$, and so $\adens(\by_1) = \dens(\by_1)$, $\adens(\by_i) = \int \dens(\by_1)  \dens(\by_i|\by_1) d\by_1 = \dens(\by_i)$ for $i>1$ (i.e., the marginal distributions are exact), and $\adens(\by_i,\by_j) = \int \dens(\by_1) \dens(\by_i|\by_1) \dens(\by_j|\by_1) d\by_1$. Hence, as for the FSA, we have $\widehat\var(\by_i)
= \var(\by_i)$, and for $i\neq j > 1$,
\begin{align*}
\widehat\cov(\by_i,\by_j) & \textstyle= \int \int \int \by_i \by_j' \dens(\by_1) \dens(\by_i|\by_1) \dens(\by_j|\by_1) d\by_i d\by_j d\by_1 \\
  & \textstyle= \int \left(\int\by_i\dens(\by_i|\by_1) d\by_i\right) \left(\int\by_j'\dens(\by_j|\by_1) d\by_j\right) \dens(\by_1) d\by_1 \\
  & \textstyle= \cov\big(\E(\by_i|\by_1),\E(\by_j|\by_1)\big),
\end{align*}
where $\E(\by_i|\by_1)$ is the predictive process with knots $\locs_1$ evaluated at $\locs_i$.
The recently proposed smoothed FSA \citep{Zhang2018}, which is billed as a generalization of the Vecchia approach, can also be viewed as a special case of general Vecchia, for which the conditioning vectors include some nearby blocks in addition to the knot vector $\by_1$.

\subsection{Multi-resolution approximation (MRA)\label{sec:mra}}

The MRA  \citep[][]{Katzfuss2015} is an iterative extension of the FSA-block, in which the domain $\domain$ is iteratively partitioned into $J$ subregions, and we select $r_i$ variables in each of the resulting subregions, such that $\locs_i \subset \domain_i$. For example, if $J=4$, let $\domain_1 = \domain$, and define $\{\domain_2,\ldots,\domain_5\}$ to be a partition of $\domain_1$, $\{\domain_6,\ldots,\domain_9\}$ to be a partition of $\domain_2$, $\{\domain_{10},\ldots,\domain_{13}\}$ to be a partition of $\domain_3$, and so forth. Set
$
q(i) = \{ j : \domain_i \subset \domain_j \},
$
and $q_y(i) = q(i)$, so that the conditioning vector consists of latent variables associated with locations above it in the hierarchy.

The FSA and MPP are special cases of the MRA. All three methods allow latent variables at unobserved locations, such that $\locs$ is different from the set of observed locations, which can be handled in our framework by $o \subset (1,\ldots,\ell)$.

\subsection{Related approach: Composite likelihood (CL)}

CL is a popular approach for fast GP inference. \citet{Varin2011} categorize CL methods as either marginal or conditional. A common marginal CL approach is pairwise blocks, which approximates the likelihood as
$
\adens(\bz) = \prod \dens(\bz_i,\bz_j),
$
where the product is often over all pairs $(i,j)$ of neighboring blocks \citep[e.g.,][]{Eidsvik2012}. 
Conditional CL is an approximation of the form \eqref{eq:approxdecomp}, except that more general conditioning index vectors $g(i) \subset (1,\ldots,i-1,i+1,\ldots,\ell)$ are considered.
In contrast to Vecchia approaches, CL-based inference is not generally guaranteed to become exact as the number of considered pairs or conditioning variables increases, and $\adens(\bz)$ is not generally guaranteed to be a valid joint density, which can make CL-based Bayesian inference difficult \citep[e.g.,][]{Shaby2014}.
While the Vecchia approaches reviewed in Section \ref{sec:vecchia} are special cases of conditional CL, 
our general Vecchia framework in \eqref{eq:vecchia1} is not a CL approach, in that  is defined on $\bx$, not on $\bz$ alone, and so it generally cannot be written in the form \eqref{eq:approxdecomp}.
Simulation studies comparing parameter estimation using Vecchia and CL approaches can be found in Appendix \ref{sec:clcomp}.

 \begin{figure}
 	\begin{subfigure}{.55\textwidth}
 	\centering
 	Standard Vecchia\\
 	\includegraphics[trim={62mm 25mm 62mm 20mm},clip,width =1\linewidth]{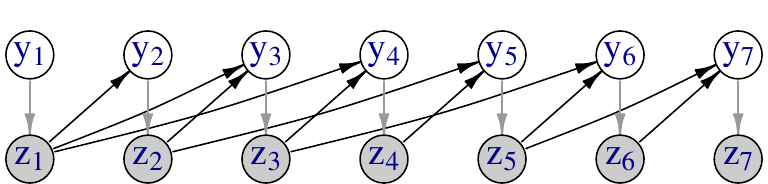} 
 	\end{subfigure}%
 \hfill
 	\begin{subfigure}{.2\textwidth}
 	\centering
 	\includegraphics[width =1\linewidth]{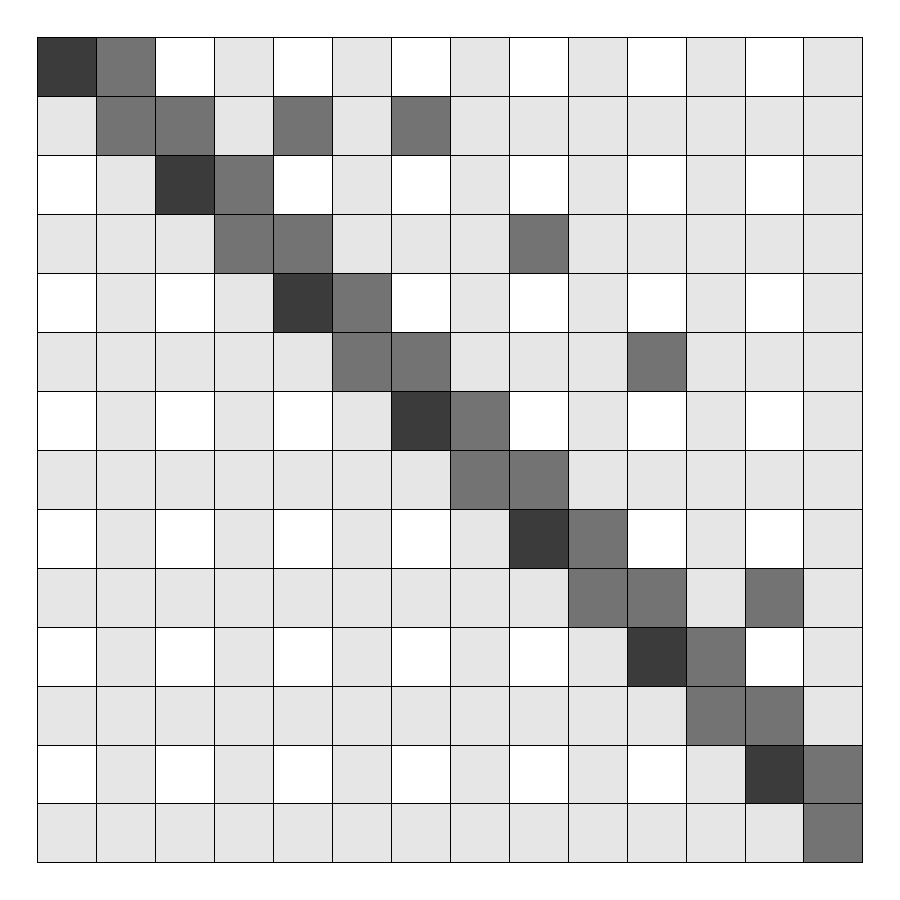}
 	\end{subfigure}%
 \hfill
 	\begin{subfigure}{.2\textwidth}
 	\centering
 	\includegraphics[width =1\linewidth]{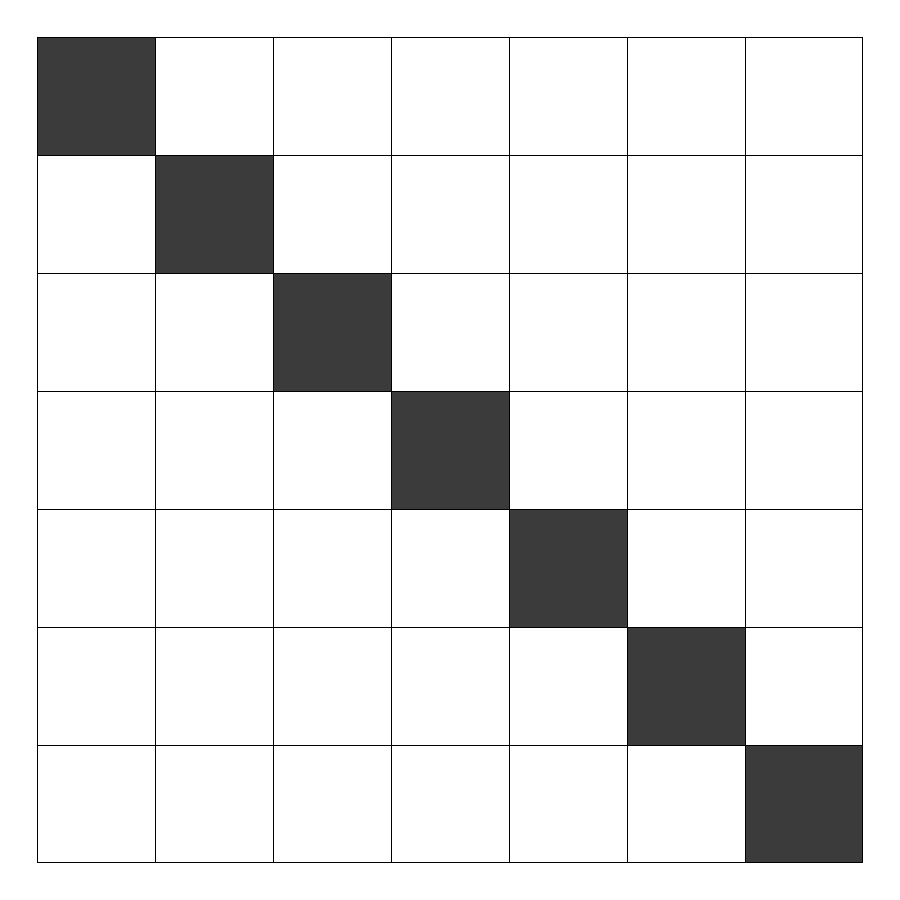}
 	\end{subfigure}

\hrulefill

 	\begin{subfigure}{.55\textwidth}
 	\centering
 	Latent Vecchia (NNGP)\\
 	\includegraphics[trim={62mm 25mm 62mm 20mm},clip,width =1\linewidth]{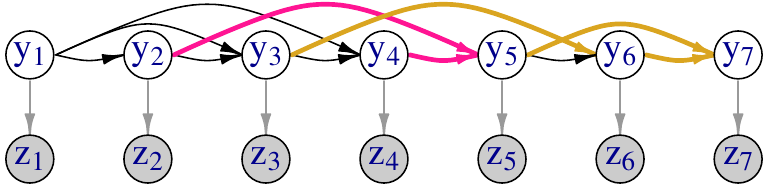} 
 	\end{subfigure}%
 \hfill
 	\begin{subfigure}{.2\textwidth}
 	\centering
 	\includegraphics[width =1\linewidth]{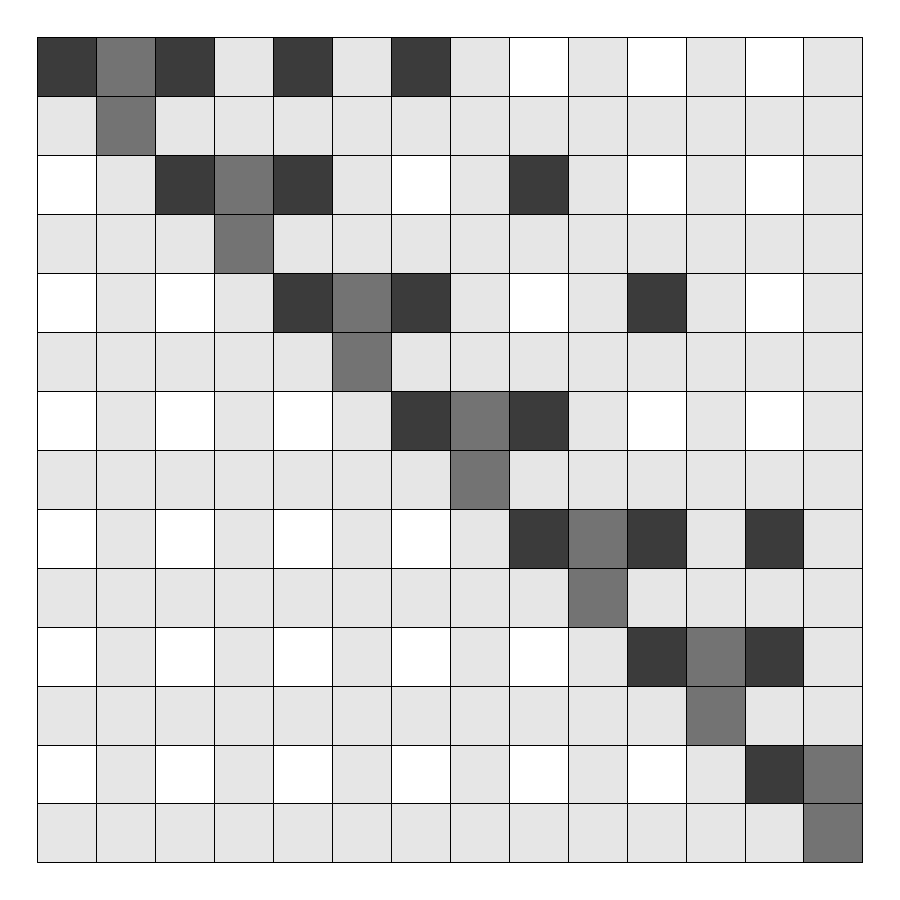}
 	\end{subfigure}%
 \hfill
 	\begin{subfigure}{.2\textwidth}
 	\centering
 	\includegraphics[width =1\linewidth]{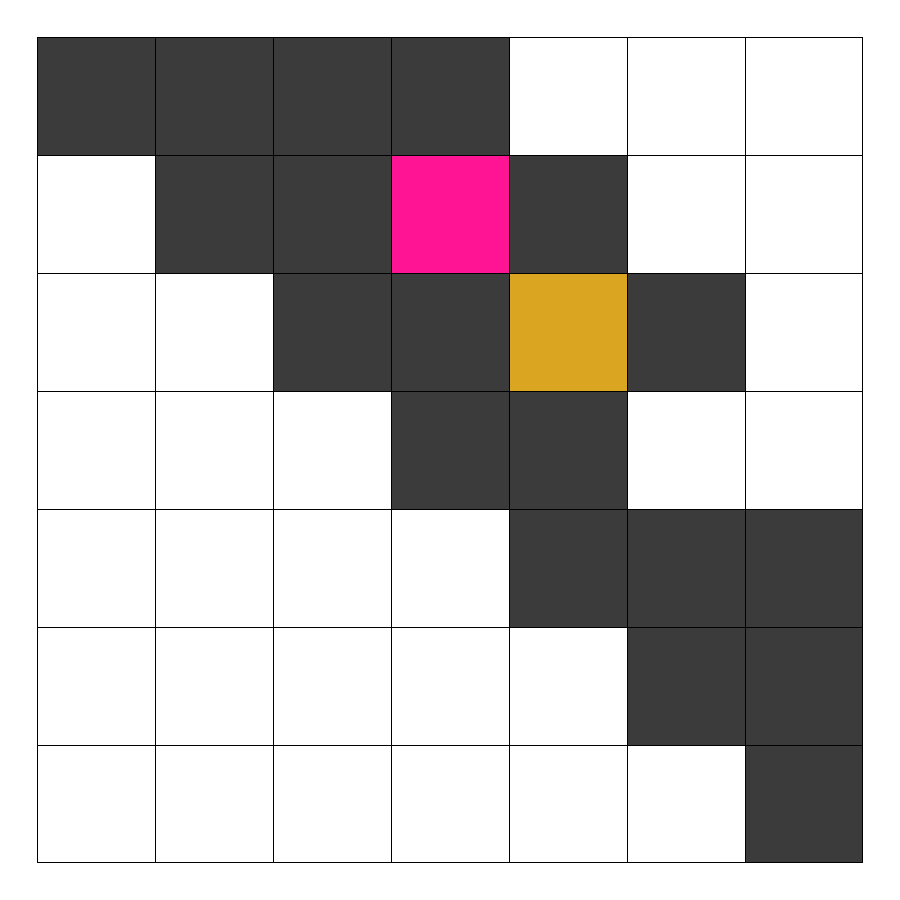}
 	\end{subfigure}

\hrulefill

 	\begin{subfigure}{.55\textwidth}
 	\centering
 	Sparse general Vecchia (SGV)\\
  	\includegraphics[trim={62mm 25mm 62mm 20mm},clip,width =1\linewidth]{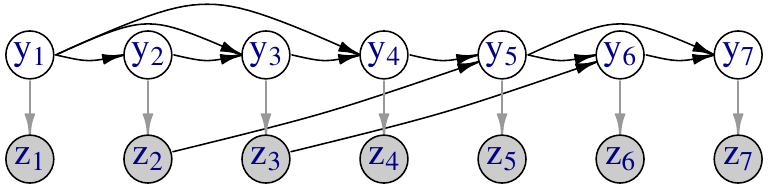} 
 	\end{subfigure}%
 \hfill
 	\begin{subfigure}{.2\textwidth}
 	\centering
 	\includegraphics[width =1\linewidth]{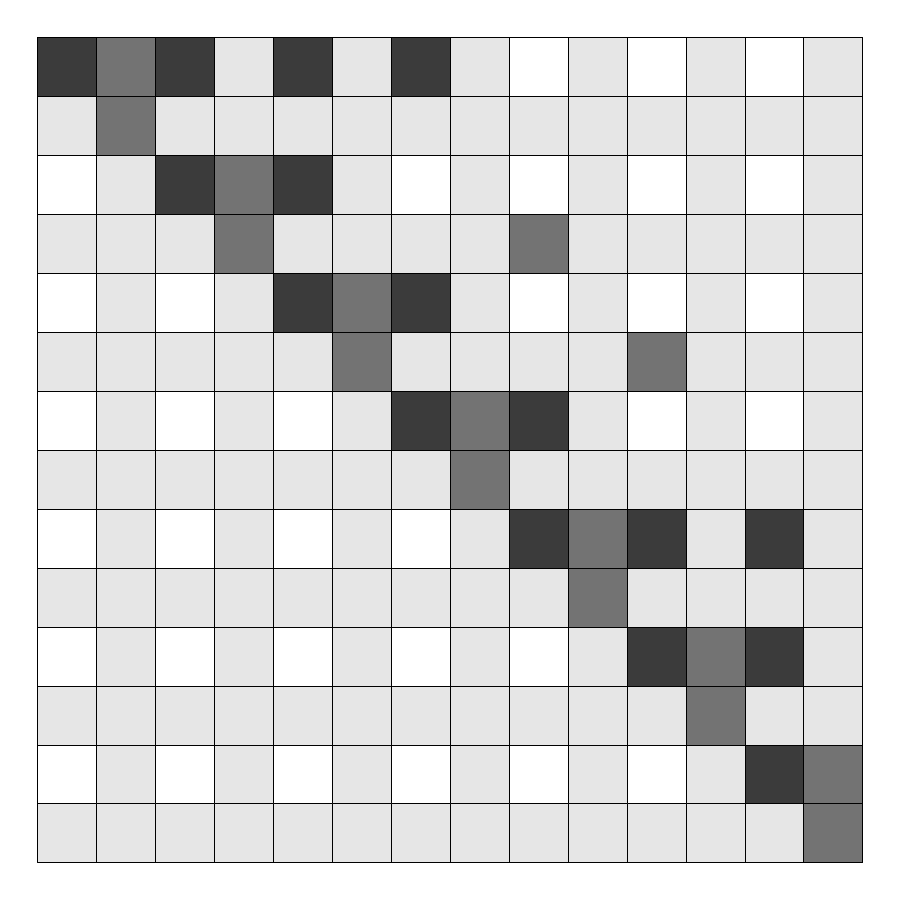}
 	\end{subfigure}%
 \hfill
 	\begin{subfigure}{.2\textwidth}
 	\centering
 	\includegraphics[width =1\linewidth]{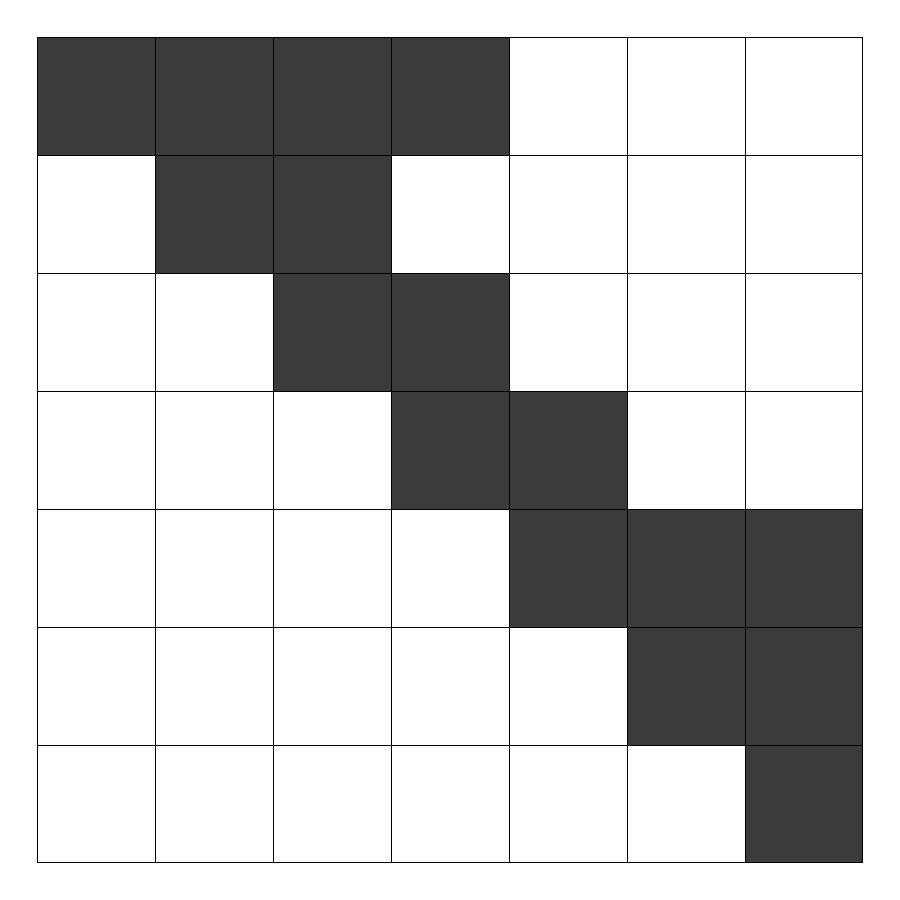}
 	\end{subfigure}

\hrulefill

 	\begin{subfigure}{.55\textwidth}
 	\centering
 	AR(2)\\
  	\includegraphics[trim={62mm 25mm 62mm 20mm},clip,width =1\linewidth]{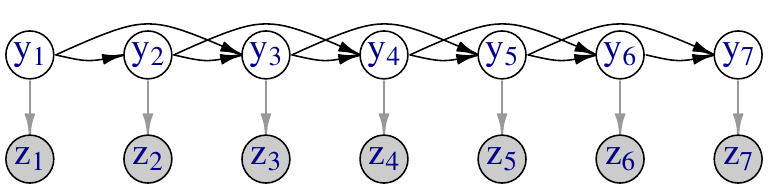} 
 	\end{subfigure}%
 \hfill
 	\begin{subfigure}{.2\textwidth}
 	\centering
 	\includegraphics[width =1\linewidth]{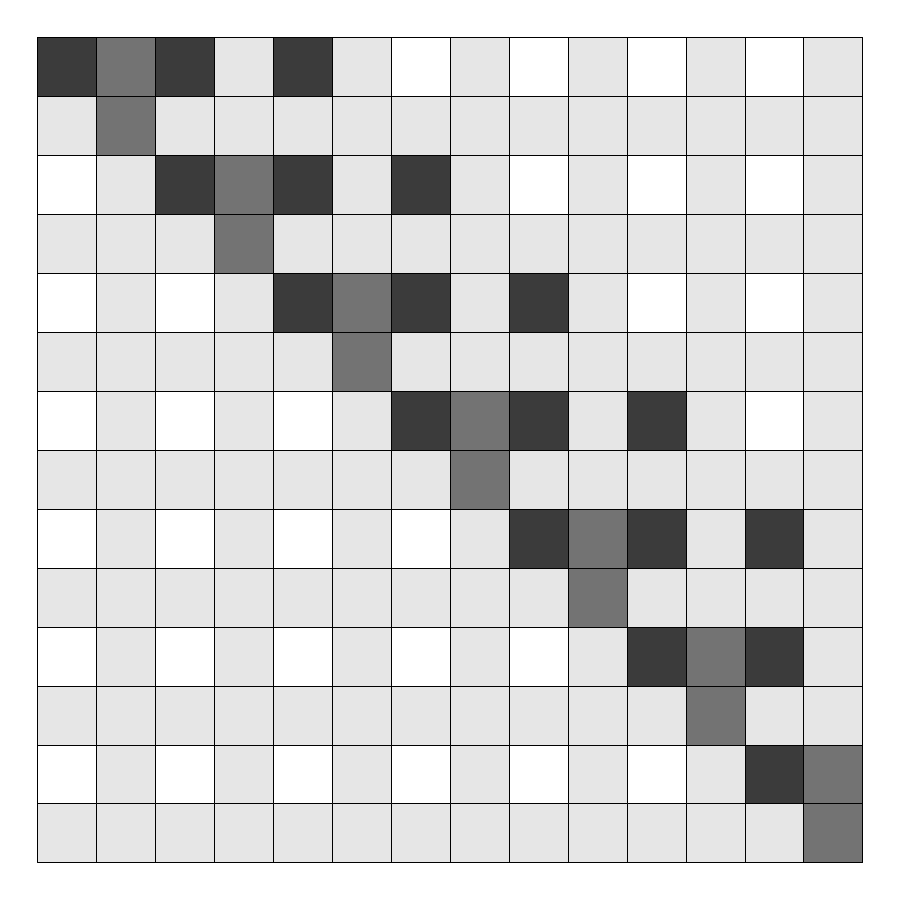}
 	\end{subfigure}%
 \hfill
 	\begin{subfigure}{.2\textwidth}
 	\centering
 	\includegraphics[width =1\linewidth]{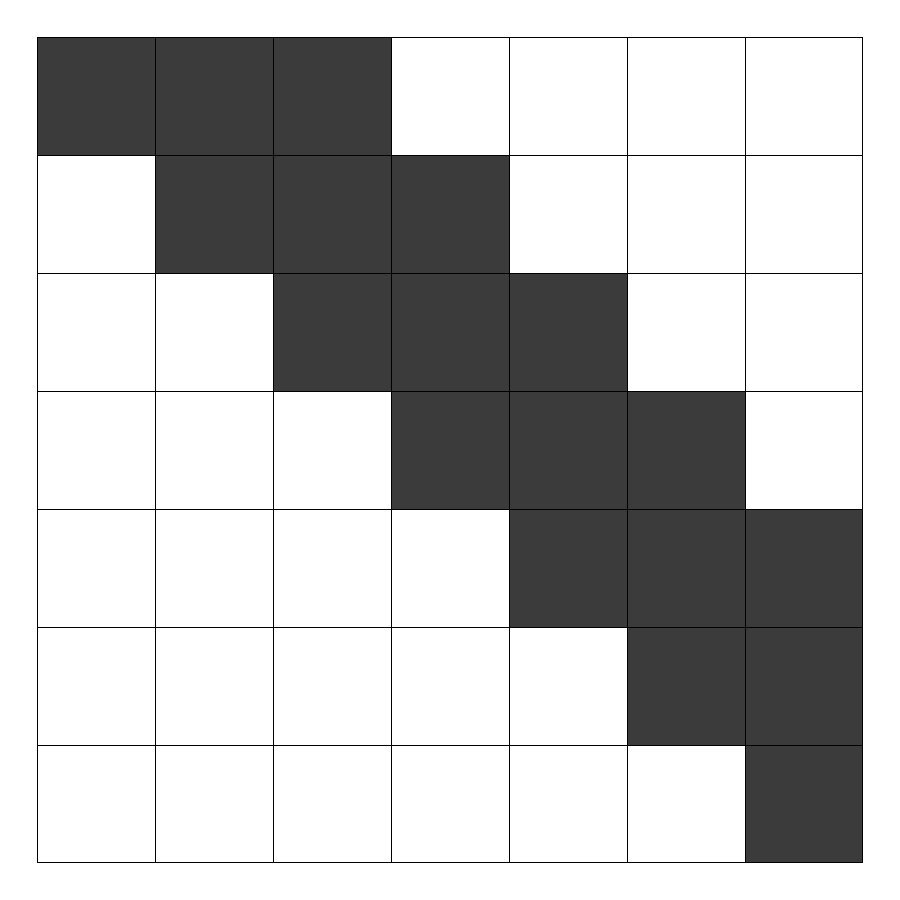}
 	\end{subfigure}

\hrulefill

 	\begin{subfigure}{.55\textwidth}
 	\centering
 	Full-scale approximation (FSA)\\
 	\includegraphics[trim={62mm 22mm 62mm 22mm},clip,width =1\linewidth]{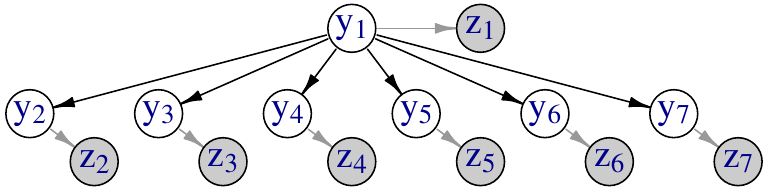} 
 	\end{subfigure}%
 \hfill
 	\begin{subfigure}{.2\textwidth}
 	\centering
 	\includegraphics[width =1\linewidth]{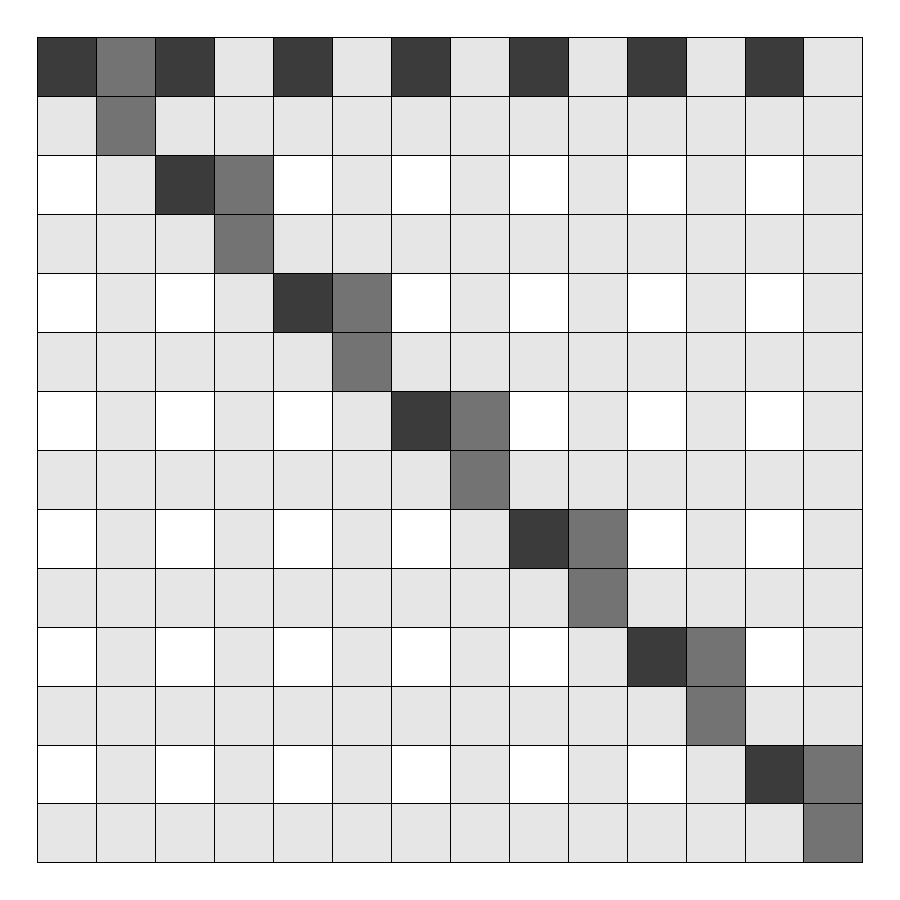}
 	\end{subfigure}%
 \hfill
 	\begin{subfigure}{.2\textwidth}
 	\centering
 	\includegraphics[width =1\linewidth]{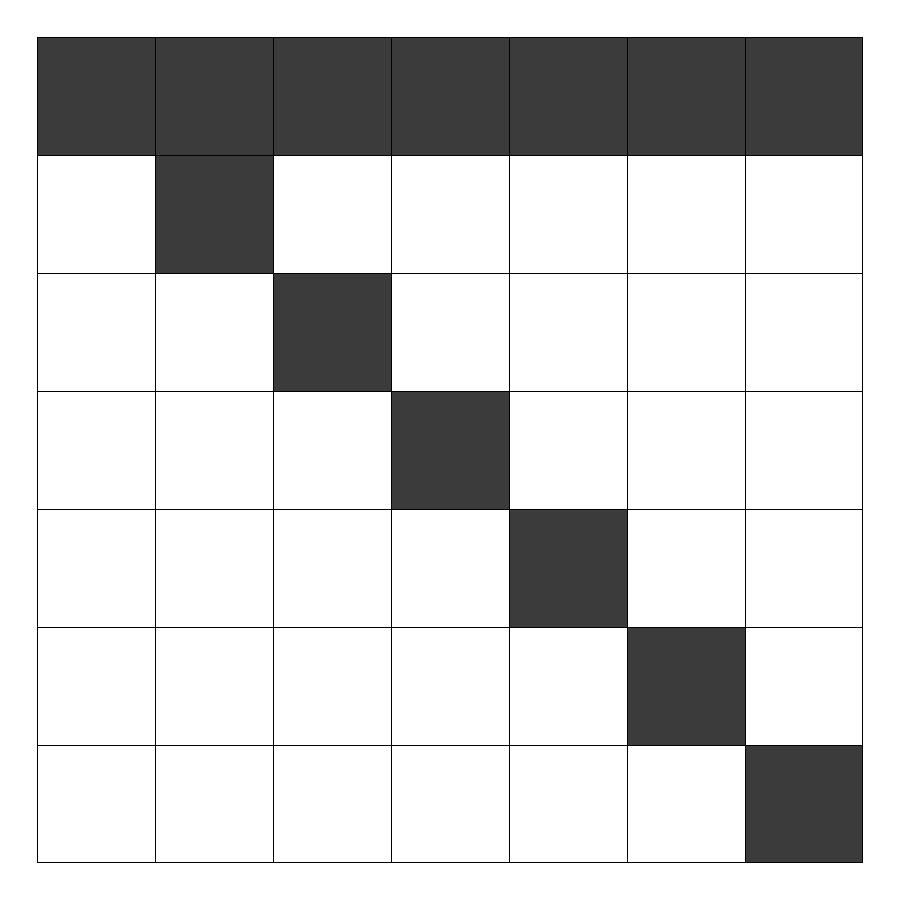}
 	\end{subfigure}

\hrulefill

 	\begin{subfigure}{.55\textwidth}
 	\centering
 	Multi-resolution approximation (MRA), $J=2$
  	\includegraphics[trim={62mm 22mm 62mm 22mm},clip,width =1\linewidth]{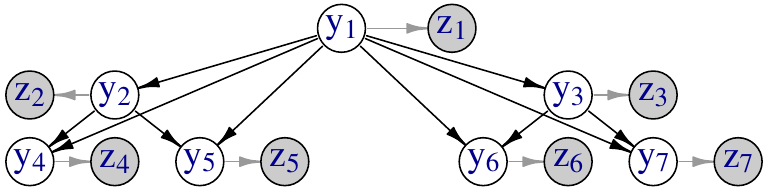} 
 	\end{subfigure}%
 \hfill
 	\begin{subfigure}{.2\textwidth}
 	\centering
 	\includegraphics[width =1\linewidth]{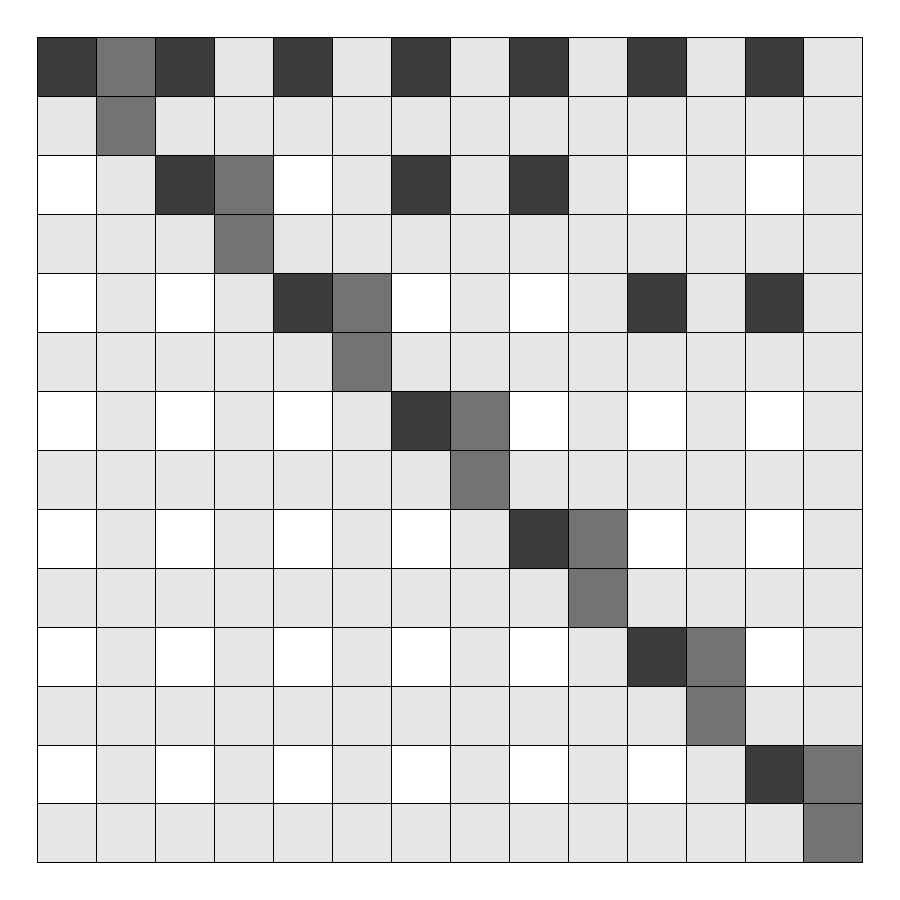}
 	\end{subfigure}%
 \hfill
 	\begin{subfigure}{.2\textwidth}
 	\centering
 	\includegraphics[width =1\linewidth]{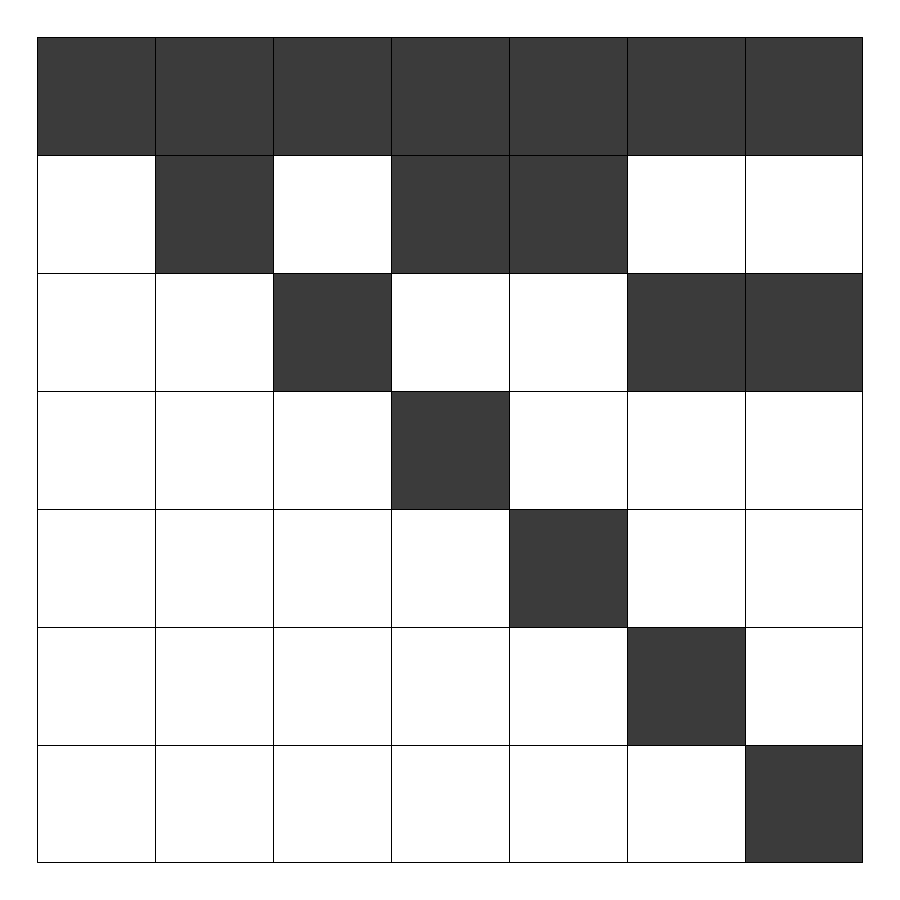}
 	\end{subfigure}%
 \caption{Toy examples of special cases of our general Vecchia approach (see Section \ref{sec:existing}) with $\ell=7$ and $o=(1,\ldots,7)$, including sparse general Vecchia (Section \ref{sec:sgv}). First column: DAGs (see Section \ref{sec:dags}). Second column: sparsity of $\bU$ (elements corresponding to $\bz_i$ in gray). Third column: sparsity of $\bV$ (Section \ref{sec:sparsity}). Computational complexity depends on the number of off-diagonal nonzeros in each column of $\bU$ and $\bV$ (Section \ref{sec:complexity}). For all methods except latent Vecchia, these numbers are at most $m=2$. For latent Vecchia, two nonzero-producing paths and the resulting nonzeros are highlighted (see Section \ref{sec:sgv}).}
 \label{fig:toyillustration}
 \end{figure}

\section{Inference and computations \label{sec:computation}}

In this section, we describe matrix representations of general Vecchia approximations, which enable fast inference. Further, we examine the sparsity of the involved matrices and derive the computational complexity.

\subsection{Matrix representations of general Vecchia}

For any two subvectors $\bx_i$ and $\bx_j$ of $\bx$, we write $C(\bx_i,\bx_j) = E(\bx_i \bx_j')$, the cross-covariance between $\bx_i$ and $\bx_j$. This gives $C(\by_i,\by_j) = C(\bz_i,\by_j) = K(\locs_i,\locs_j)$, $C(\bz_i,\bz_j) = K(\locs_i,\locs_j)$ for $i\neq j$, and $C(\bz_i,\bz_i) = K(\locs_i,\locs_i) + \tau^2 \bI_{r_i}$.

We can write the general Vecchia approximation in \eqref{eq:vecchia1} as
\begin{equation}
\label{eq:densproduct}
\adens(\bx) = \prod_{i=1}^b \dens(\bx_i | \bx_{g(i)}) = \prod_{i=1}^b \normal(\bx_i | \bB_i\bx_{g(i)},\bD_i \big)
\end{equation}
where $\bB_i= C(\bx_i,\bx_{g(i)}) C(\bx_{g(i)},\bx_{g(i)})^{-1}$ and $\bD_i = C(\bx_i,\bx_i) - \bB_i C(\bx_{g(i)},\bx_i)$. We view $\bB_i$ as a block matrix, and so $(\bB_i)_{\#(j,g(i))}$ is the block of $\bB_i$ corresponding to $\bx_j$ when $j \in g(i)$.

For any symmetric, positive-definite matrix $\bA$, let $\chol(\bA)$ be the lower-triangular Cholesky factor of $\bA$, and let $\bP$ be a permutation matrix so that $\bP \bA$ reorders the rows of matrix $\bA$ in reverse order. Then, we call $\rchol(\bA) \colonequals \bP(\chol(\bP\bA\bP))\bP$ the reverse Cholesky factor of $\bA$ (i.e., the row-column reversed Cholesky factor of the row-column reversed $\bA$). The following proposition is a standard result for the multivariate normal distribution:
\begin{prop}
\label{prop:matrixrep}
For the density in \eqref{eq:densproduct}, we have $\adens(\bx) = \normal_n(\bx | \bfzero, \widehat\bC)$, where $\widehat \bC^{-1} = \bU \bU'$, $\bU$ is a sparse upper triangular $b \times b$ block matrix with $(j,i)$th block
\begin{equation}
\label{eq:U}
\bU_{ji} = \begin{cases} \bD_i^{-1/2}, & i=j,\\
-(\bB_{i})_{\#(j,g(i))}'\,\bD_i^{-1/2}, & j \in g(i), \\
\bfzero, &\textnormal{otherwise}, \end{cases}
\end{equation}
and $\bD_i^{-1} = \bD_i^{-1/2}(\bD_i^{-1/2})'$. Further, $\bU = \rchol(\widehat\bC^{-1})$ is the reverse Cholesky factor of $\widehat\bC^{-1}$.
\end{prop}
All proofs can be found in Appendix \ref{app:proofs}.


\subsection{Likelihood \label{sec:likelihood}}

By integrating $\adens(\bx)$ with respect to the latent $\by$, the general Vecchia approximation implies a distribution for the observed vector $\bz_o$ as in \eqref{eq:induceddist}. For large $n_y$, numerical integration with respect to the $n_y$-dimensional vector $\by$ is challenging \citep[see][]{Finley2017}. Hence, we consider the analytically integrated density instead:
\begin{prop}
\label{prop:likelihood}
The general Vecchia likelihood can be computed as:
\begin{equation}
\label{eq:likelihood}
-2 \log \adens(\bz_o) = \sum_{i=1}^b \log |\bD_i| + 2\sum_{i=1}^\ell \log |\bV_{ii}| + \tilde\bz'\tilde\bz - (\bV^{-1}\bU_Y \tilde\bz)'(\bV^{-1}\bU_Y\tilde\bz) + n_z \log(2\pi),
\end{equation}
where $\bV \colonequals \rchol(\bW)$, $\bW \colonequals \bU_Y \bU_Y'$, $\tilde\bz \colonequals \bU_Z'\bz_o$,
and $\bU_Y \colonequals \bU_{\#( \by, \bx )\,\all}$ and $\bU_Z \colonequals \bU_{\#( \bz_o, \bx )\,\all}$ consist of the rows of $\bU$ corresponding to $\by$ and $\bz_o$, respectively.
\end{prop}
Note that, analogously to $\bU = \rchol( \widehat\bC^{-1})$ in Proposition \ref{prop:matrixrep}, we compute $\bV = \rchol(\bW)$ as the reverse Cholesky factor of $\bW$. This allows us to derive the sparsity structure of $\bV$ in Proposition \ref{prop:sparsity} below, and ensures low computational complexity for certain configurations of general Vecchia.

Thus, for these configurations, the likelihood $\adens(\bz_o)$ (integrated over $\by$) can be evaluated quickly for any given value of the parameters $\bftheta$ and $\tau^2$, which enables likelihood-based parameter inference even for very large datasets. Our framework is agnostic with respect to the inferential paradigm, allowing both frequentist and Bayesian inference.
Frequentist inference can be carried out by finding the parameter values that maximize $\adens(\bz_o)$. Appendix \ref{sec:clcomp}
details a simulation in which we compared SGV to the exact likelihood and two composite likelihood methods on estimating a spatial range
parameter. We found that SGV gave very similar parameter estimates to the exact maximum likelihood estimates. 
Appendix \ref{sec:bayesian} demonstrates how the SGV likelihood can be used to conduct Bayesian inference. This example considered both numerically integrated posteriors and a Metropolis-Hastings algorithm, based on which we obtained the posterior predictive distribution of $y(\cdot)$ at unobserved locations.

No matter the inferential paradigm, our models can be viewed as approximations of GP models, or as valid probability models in their own right (see \eqref{eq:densproduct}). All our inference is exact from the latter perspective. The error due to the Vecchia approximation itself disappears for large $m=n-1$, and it is examined numerically for smaller $m$ in Section \ref{sec:numerical}.


\subsection{Prediction \label{sec:prediction}}

For prediction, we can compute the posterior distribution of the error-free process vector given by $\by | \bz \sim \normal(\bfmu,\bW^{-1})$, where $\bfmu \colonequals -\bW^{-1} \bU_y\tilde\bz$. However, this requires consideration of complex issues, such as how to guarantee fast computation of relevant summaries of this distribution, and what ordering and conditioning strategies work well in the context of prediction at observed and unobserved locations. Thus, we refer to \citet{Katzfuss2018} for details on how to extend the general Vecchia framework to GP prediction.

\subsection{Sparsity structures \label{sec:sparsity}}

In the following proposition, we use connections between Vecchia approaches and DAGs (see Section \ref{sec:dags} and Appendix \ref{app:dags}) to verify the sparsity structure of $\bU$ in \eqref{eq:U} and to determine the sparsity of $\bW$ and $\bV = \rchol(\bW)$, which must be computed for inference.
\begin{prop}
\label{prop:sparsity} $ $
\begin{enumerate}[itemsep=1pt,topsep=1.5pt,parsep=1pt]
\item $\bU_{ji} = \bfzero$ if $i \neq j$ and $\bx_j \not\to \bx_i$.
\item $\bW_{ji} = \bfzero$ if $i\neq j$, $\by_j \not\to \by_i$, $\by_i \not\to \by_j$, and there is no $k> \max(i,j)$ such that both $\by_i \to \by_k$ and $\by_j \to \by_k$.
\item $\bV_{ji} = \bfzero$ if $j>i$. For $j<i$, $\bV_{ji}=\bfzero$ if there is no path between $\by_i$ and $\by_j$ on the subgraph $\{\by_i,\by_j\} \cup \{\by_k: k>i, \by_k \textnormal{ has at least one observed descendant}\}$.
\end{enumerate}
\end{prop}
Thus, $\bU$ and $\bV$ are upper triangular, and the sparsity of the upper triangle depends on the Vecchia specification. For $j<i$, the $(j,i)$ block of $\bW$ is not only nonzero if $j \in q_y(i)$, but also if $\by_i$ and $\by_j$ appear in a conditioning vector together (i.e., $i,j \in q_y(k)$ for some $k$). Note that this sparsity structure corresponds to the adjacent vertices in the so-called moral graph \citep[e.g.,][Sect.~2.1.1]{Lauritzen1996}. The matrix $\bV$ is typically at least as dense as $\bW$ (assuming that all or most $\by_k$ have observed descendants), in that it has the same nonzeros as $\bW$ plus additional ones induced by more complicated paths. Figure \ref{fig:toyillustration} shows examples of DAGs with the corresponding sparsity structures of $\bU$ and $\bV$.

\subsection{Computational complexity \label{sec:complexity}}

Recall that $\bx$ consists of $\by=(\by_1,\ldots,\by_\ell)$ and $\bz_o$, where $\bz=(\bz_1,\ldots,\bz_\ell)$, and $\by_i$ and $\bz_i$ are of length $r_i$. Let $n$ be the total number of individual variables in $\bx$.
To simplify the sparsity and computational complexity calculations, assume that $o = (1,\ldots,\ell)$, all $r_i$ are of the same order $r$ (and so $n \approx br = 2\ell r$), and that all conditioning vectors consist of at most $m$ subsets (of size $r$): $|g(i)| \leq m$.

Then, it is easy to see from \eqref{eq:U} that $\bU$ has $\order(b\cdot(mr^2)) = \order(nmr)$ nonzero elements and can be computed in $\order(b\cdot(m^3r^3)) = \order(nm^3r^2)$ time.
Note that this time complexity is lower in $r$ than in $m$, whose product $mr$ makes up the total length of the conditioning vectors.

To evaluate the likelihood in \eqref{eq:likelihood}, we also need to compute $\bW = \bU_Y\bU_Y'$ and find its reverse Cholesky factor $\bV = \rchol(\bW)$.
Each conditioning vector is of size $|g(i)| \leq m$ and so contains at most $m(m-1)/2$ pairs of elements. Therefore, from Proposition \ref{prop:sparsity}, $\bW$ has at most $\order(\ell m^2)$ non-zero blocks. Since each block is of size $r \times r$ and $\ell = \order(n/r)$, $\bW$ has at most $\order(nrm^2)$ non-zero elements.

The time complexity for obtaining a lower-triangular Cholesky factor is on the order of the sum of the squares of the number of nonzero elements per column in the factor \citep[e.g.,][Thm.~2.2]{Toledo2007}. It can be easily verified that the same holds for our reverse Cholesky factor $\bV$. Thus, for any particular Vecchia approximation, the time complexity for inference can in principle be determined based on the corresponding DAG using Proposition \ref{prop:sparsity}. We give examples in the next section.


\section{Sparse general Vecchia (SGV) approximation \label{sec:sgv}}

We now study choice C5 from Section \ref{sec:generalvecchia} for fixed choices C1--C4; that is, we assume that the grouping, ordering, and the conditioning index vectors $q(1),\ldots,q(\ell)$ are fixed. We consider three methods that differ in their choice of latent versus observed conditioning (C5): the existing methods standard Vecchia (Section \ref{sec:vecchia}) and latent Vecchia (used in the NNGP in Section \ref{sec:nngp}), and a novel sparse general Vecchia (SGV) approach:
\begin{description}[itemsep=0.5pt]
\item[Standard Vecchia ($\adens_s$):] $q_z(i) = q(i)$, condition only on observed vectors $\bz_j$.
\item[Latent Vecchia ($\adens_l$):] $q_y(i) = q(i)$, condition only on latent vectors $\by_j$.
\item[SGV ($\adens_g$):] For each $i$, partition $q(i)$ into $q_y(i)$ and $q_z(i)$ such that $j$ and $k$ with $j<k$ can only both be in $q_y(i)$ if $j \in q_y(k)$.
\end{description}
In the terminology of \citet[][Sect.~2.1.1]{Lauritzen1996}, SGV ensures that the corresponding DAG forms a perfect graph. Different versions of SGV are possible for the same conditioning index vectors (and, in fact, standard Vecchia is one special case of SGV). Throughout this article, 
we consider the following strategy that attempts to maximize latent conditioning in the SGV: We obtain the latent-conditioning index vector $q_y(i)$ for each $i=2,\ldots,\ell$ by first finding the index $k_i \in q(i)$ whose latent-conditioning index vector has the most overlap with $q(i)$: $k_i=\arg\max_{j \in q(i)} |q_y(j) \cap q(i)|$. In case of a tie, we choose the $k_i$ for which the spatial distance between $\locs_i$ and $\locs_{k_i}$ is shortest. Then, we set $q_y(i) = (k_i) \cup (q_y(k_i) \cap q(i))$, with the remaining indices in $q(i)$ corresponding to observed conditioning: $q_z(i) = q(i) \setminus q_y(i)$.

The three approaches are illustrated in a toy example with $\ell=7$ shown in Figure \ref{fig:toyillustration}. For all three methods we have the same $q(1),\ldots,q(\ell)$: $q(2) = (1)$, $q(3) = ( 1,2 )$, $q(4)=(1,3)$, $q(5)=(2,4)$, \ldots. Like latent Vecchia, SGV uses $q_y(2) = (1)$, $q_y(3) = ( 1,2 )$, $q_y(4)=(1,3)$, as, for example, $1 \in q_y(3)$, and so $q_y(4)$ can contain both $1$ and $3$. However, $2 \notin q_y(4)$, and so SGV does not allow both $2 \in q_y(5)$ and $4 \in q_y(5)$, and sets $q_y(5)=(4)$ and $q_z(5)=(2)$.

We now establish an ordering on the accuracy of the approximations to $\dens(\bx)$.
\begin{prop}
\label{prop:KLordering}
The following ordering of Kullback-Leibler (KL) divergences holds:
\[\textnormal{KL}\big(\dens(\bx) \| \adens_l(\bx)\big) \leq \textnormal{KL}\big(\dens(\bx) \| \adens_g(\bx)\big) \leq \textnormal{KL}\big(\dens(\bx) \| \adens_s(\bx)\big).\]
\end{prop}
Thus, the approximation accuracy for the joint distribution of $\bx$ is better for latent Vecchia than for SGV, which is better than that for standard Vecchia. Note, however, that this does not guarantee that the KL divergence for the implied distribution of the observations $\bz_o$ follows the same ordering.
For example, Proposition \ref{prop:KLordering} says that $\E^\dens(\log \adens_g(\bx)) \geq \E^\dens(\log \adens_s(\bx))$, but that does not guarantee that $\E^\dens(\log \int \adens_g(\bx) d\by) \geq \E^\dens(\log \int \adens_s(\bx) d\by)$. Examples of this can be found in Figure \ref{subfig:kl2dsnr1exp}.

Another important factor is the computational complexity of the different approaches. Standard Vecchia only conditions on observed quantities, and so for any $i,j$ we have $\by_j \not\to \by_i$, resulting in a diagonal $\bW$ and $\bV$ according to Proposition \ref{prop:sparsity}, and hence an overall time complexity of $\order(nm^3r^2)$ for standard Vecchia.

\citet{Finley2017} observed numerically that matrices in the NNGP (which uses the latent Vecchia approach) were less sparse than in standard Vecchia. We can examine this issue further using Proposition \ref{prop:sparsity}. In the toy example in Figure \ref{fig:toyillustration}, latent Vecchia uses $q_y(5)=(2,4)$, which creates the path $(\by_2,\by_5,\by_4)$ that leads to $\bV_{2,4} \neq \bfzero$. Setting $q_y(6) = (3,5)$ creates the path $(\by_3,\by_6,\by_7,\by_5)$, leading to $\bV_{3,5}\neq \bfzero$. This results in $3>m=2$ nonzero off-diagonal elements in columns 4 and 5. We provide more insight into the increased computational cost for latent Vecchia in the following proposition:
\begin{prop}
\label{prop:latentcomplexity}
Consider the latent Vecchia approach with $r_i = 1$, $o=(1,\ldots,\ell)$, $m \leq n_z^{1/d}$, and coordinate-wise ordering for locations on an equidistant grid in a $d$-dimensional hypercube with nearest-neighbor conditioning. Then, $\bV = \rchol(\bW)$ has $\order(n^{1-1/d}m^{1/d})$ nonzero elements per column, requiring $\order(n^{2-1/d}m^{1/d})$ memory. The resulting time complexity for computing $\bV =\rchol(\bW)$ is $\order(n^{3-2/d} m^{2/d})$.
\end{prop}
Thus, the time complexity for obtaining $\bV$ is $\order(nm^2)$ in $d=1$ dimensions, $\order(n^2m)$ for $d=2$, and approaching the cubic complexity in $n$ of the original GP as $d$ increases. For irregular observation locations, we expect roughly similar scaling if the locations can be considered to have been drawn from independent uniform distributions over the domain. Also note that using reordering algorithms for the Cholesky decomposition (as opposed to simple reverse ordering) could lead to different complexities, although our numerical results indicate that this might actually increase the computational complexity (see Figure \ref{fig:complexity}).

In contrast to latent Vecchia, SGV results in guaranteed sparsity. In the toy example, SGV sets $q_y(5)=(4)$ and $q_z(5)=(2)$ because $2 \notin q_y(4)$, and $q_y(6) = (5)$ and $q_z(6) =(3)$ because $3 \notin q_y(5)$, resulting in $\bV_{2,4} = \bV_{3,5} = \bfzero$ (in contrast to latent Vecchia). More generally, SGV preserves the linear scaling of standard Vecchia, in any spatial dimension and for gridded or irregularly spaced locations:
\begin{prop}
\label{prop:complexity}
For SGV, $\bV$ has at most $mr$ off-diagonal elements per column, and so the time complexity for computing $\bV = \rchol(\bW)$ is only $\order(nm^2r^2)$. Thus, SGV has the same overall computational complexity as standard Vecchia.
\end{prop}

In summary, SGV provides improvements in approximation accuracy over standard Vecchia (Proposition \ref{prop:KLordering}) while retaining linear computational complexity in $n$ (Proposition \ref{prop:complexity}). Latent Vecchia results in improved approximation accuracy but can raise the computational complexity severely (Proposition \ref{prop:latentcomplexity}), which can be infeasible for large $n$.
Numerical illustrations of these results can be found in Section \ref{sec:numerical}.

\section{Ordering and conditioning \label{sec:ordcond}}

We now provide some insight into choices C2--C4 of Section \ref{sec:generalvecchia}. For simplicity, we henceforth assume $r_i = 1$ unless stated otherwise.

\subsection{Ordering (C3) \label{sec:ordering}}

In one spatial dimension, a ``left-to-right'' ordering of the locations in $\locs$ is natural. However, in two or more spatial dimensions, it is not obvious how the locations should be ordered. For Vecchia approaches, the default and most popular ordering is along one of the spatial coordinates (coord ordering). \citet{Datta2016} only observed a negligible effect of the ordering on the quality of the Vecchia approximation, but \citet{Guinness2016a} showed that this is not always the case. He proposed different ordering schemes, including an approximate maximum-minimum-distance (maxmin) ordering, which sequentially picks each location in the ordering by aiming to maximize the distance to the nearest of the previous locations. \citet{Guinness2016a} showed that maxmin ordering can lead to substantial improvements over coord ordering in settings without any nugget or noise. We will examine the nonzero nugget case in Section \ref{sec:numerical}. Note that the MRA (Section \ref{sec:mra}) implies an ordering scheme similar to maxmin, starting with a coarse grid over space and subsequently getting denser and denser.

\subsection{Choosing $m$ \label{sec:m}}

For a given ordering, as part of C4 we must choose $m$, the size of the conditioning vectors.

For one-dimensional spatial domains, some guidance can be obtained for approximating a GP with a Mat\'ern covariance on a one-dimensional domain. If the smoothness is $\nu = 0.5$, we have a Markov process of order 1, and so we can get an exact approximation for latent conditioning with $m=1$ by ordering from left to right. \citet{Stein2011} conjectures that for smoothness $\nu$, approximate screening holds for any $m>\nu$. This conjecture is explored numerically in Section \ref{sec:numerical}, specifically in Figure \ref{fig:1D_increasingSmoothness}. Note that coord ordering in 1-D with $m$-nearest-neighbor conditioning amounts to an AR($m$) model, and the corresponding latent or SGV inference is equivalent to a Kalman filter and smoother \citep[cf.][]{Eubank2002}. For very smooth processes (i.e., very large $\nu$), the $m$ necessary for (approximate) screening won't be affordable any more, and alternative ordering and conditioning strategies might be advantageous (see Section \ref{sec:conditioning} below).

For two or more dimensions, the necessary $m$ will depend not only on the smoothness of the covariance function, but also on the chosen ordering, the observation locations (regular or irregular), and other factors. We suggest starting with a relatively small $m$ and gradually increasing it using warm-starts based on previously obtained parameter estimates, until the estimates have converged to a desired tolerance, or until the available computational resources have been exhausted.

\subsection{Conditioning\label{sec:conditioning}}

For a given ordering and $m$, the most common strategy is to simply condition on the $m$ nearest neighbors or locations (NN conditioning), although more elaborate conditioning schemes have been proposed \citep[e.g.,][]{stein2004,Gramacy2015}. It can also be advantageous in some situations to place a coarse grid over space at the beginning of the ordering and to always condition on this grid.
Properties of this same-conditioning-set (SCS) approach are described in Appendix \ref{app:scs}.


\section{Numerical study \label{sec:numerical}}

We examined numerically the propositions and claims made in previous sections. We explored C3--C5 from Section \ref{sec:generalvecchia}, with an emphasis on C5 by comparing the three approaches from Section \ref{sec:sgv}. Throughout this section, we set $r_i=1$ and $o = (1,\ldots,\ell)$, so that each latent variable had a corresponding observed variable. The observation locations were equidistant grids on the unit interval or unit square, and the true GP was assumed to have Mat\'ern covariance with variance $\sigma^2$, smoothness $\nu$, and effective range $\lambda$ (i.e., the distance at which the correlation drops to 0.05). We added noise with variance $\tau^2$, set $\sigma^2 + \tau^2 = 1$, and so the signal proportion was $\sigma^2/(\sigma^2 + \tau^2) = \sigma^2$. For example, signal proportions of 1/2 and 2/3 correspond to signal-to-noise ratios (SNRs) of 1 and 2, respectively. We considered coordinate-wise (coord) ordering and the approximate maximum-minimum-distance (maxmin) ordering of \citet{Guinness2016a}. We used nearest-neighbor (NN) conditioning for a given ordering, unless stated otherwise. Comparisons among methods are made using the Kullback-Leibler (KL) divergence between the approximate distribution $\adens(\bz)$ and the true distribution $\dens(\bz)$.

First, we assumed a one-dimensional spatial domain, $\domain= [0,1]$, and only considered the natural coord ordering ``from left to right.'' The different methods from Section \ref{sec:sgv} then essentially correspond to latent or non-latent AR($m$) processes. From the results shown in Figure \ref{fig:1D} with $n_z = 100$ and $\lambda =0.9$, we can see that latent Vecchia and the equivalent SGV performed much better than standard Vecchia. Figure \ref{fig:1D_increasingSmoothness} also confirms numerically the conjecture from Section \ref{sec:m} that (approximate) screening holds for latent Vecchia if $m  > \nu$.

\begin{figure}
	\begin{subfigure}{.48\textwidth}
	\centering
	\includegraphics[width =.7\linewidth]{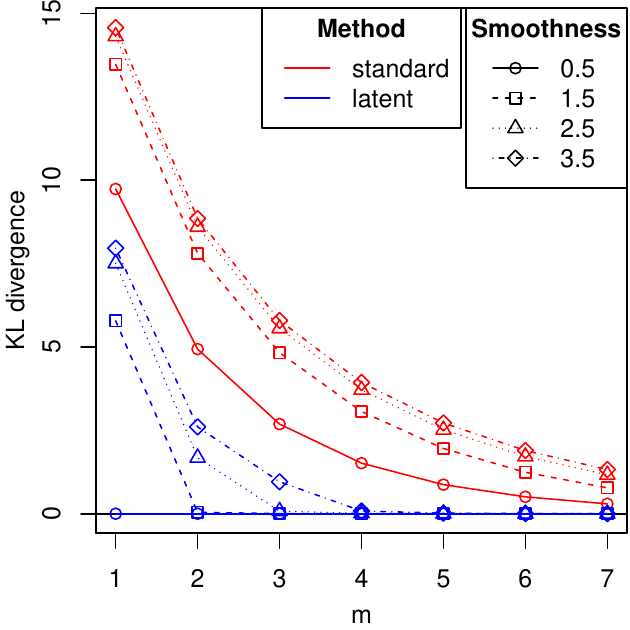}
	\caption{SNR $=1$}
	\label{fig:1D_increasingSmoothness}
	\end{subfigure}%
\hfill
	\begin{subfigure}{.48\textwidth}
	\centering
	\includegraphics[width =.7\linewidth]{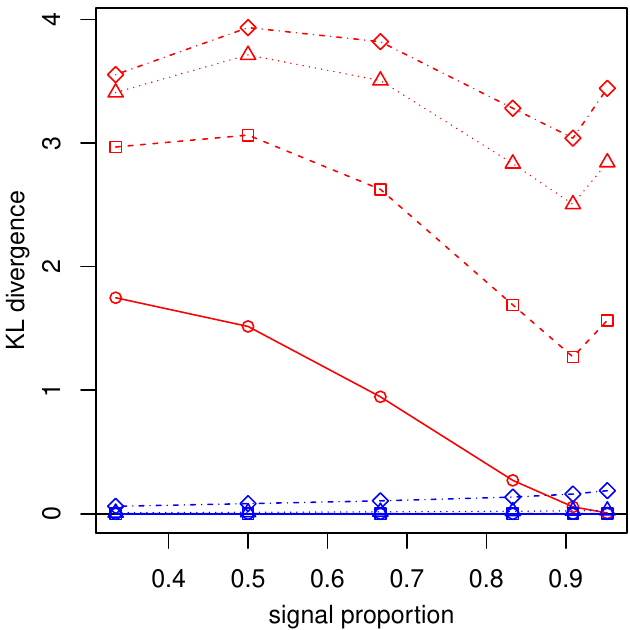}
	\caption{$m=4$}
	\end{subfigure}%
  \caption{KL divergences for Vecchia approximations of a GP with Mat\'ern covariance on the unit interval with coord ordering and NN conditioning. SGV is equivalent to latent Vecchia in this setting.}
\label{fig:1D}
\end{figure}

\begin{figure}[!htb]
\centering\includegraphics[width =.75\linewidth]{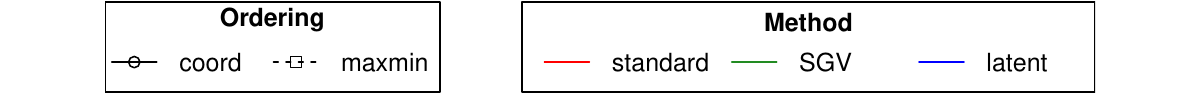}

\vspace{2mm}
	\begin{subfigure}{.33\textwidth}
	\centering
	\includegraphics[width =.97\linewidth]{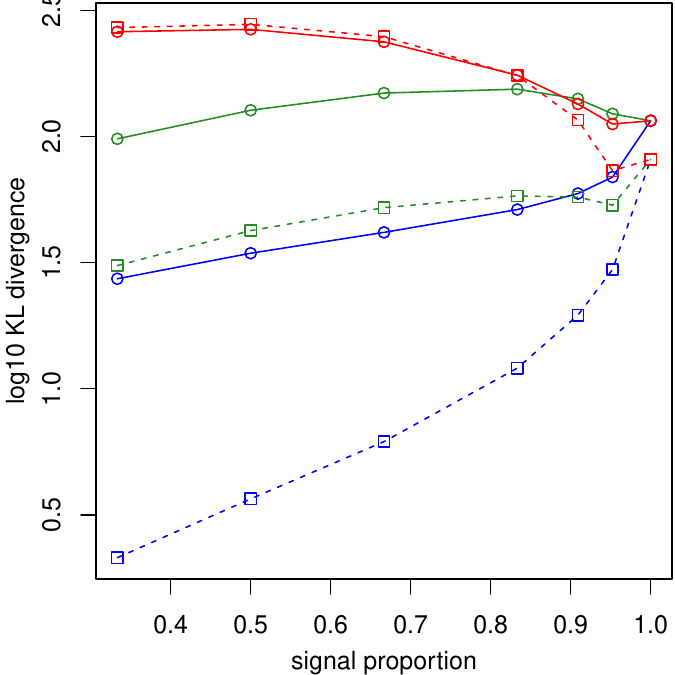}
	\caption{$m=5, \; \nu = 0.5$}
	\end{subfigure}%
\hfill
	\begin{subfigure}{.33\textwidth}
	\centering
	\includegraphics[width =.97\linewidth]{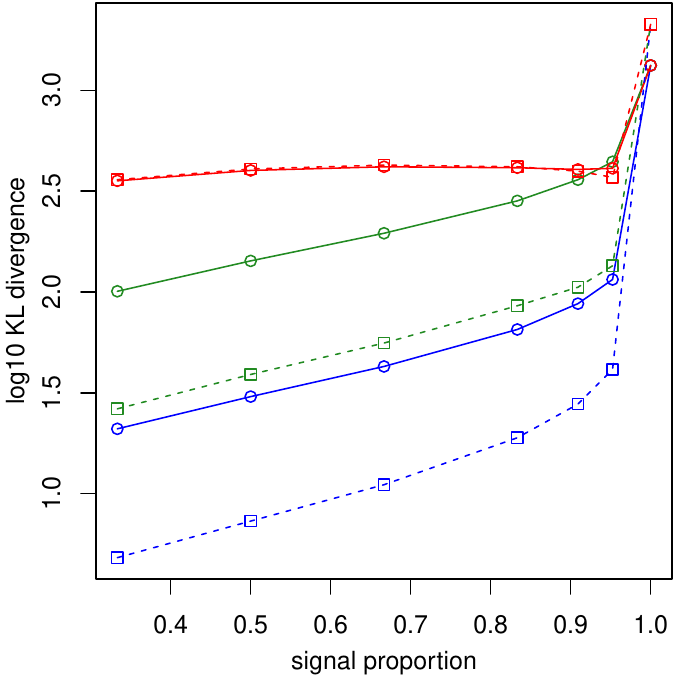}
	\caption{$m=5, \; \nu = 1.5$}
	\end{subfigure}%
\hfill
	\begin{subfigure}{.33\textwidth}
	\centering
	\includegraphics[width =.97\linewidth]{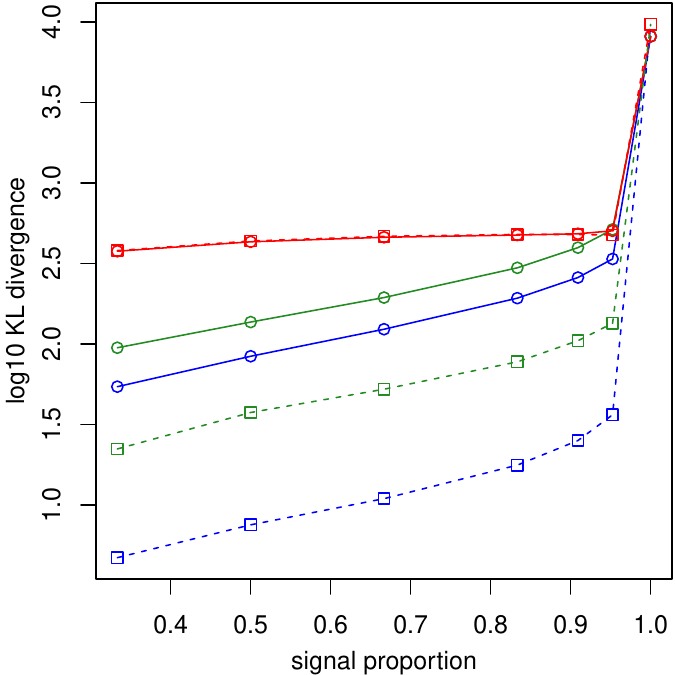}
	\caption{$m=5, \; \nu = 2.5$}
	\end{subfigure}%

\vspace{3mm}
	\begin{subfigure}{.33\textwidth}
	\centering
	\includegraphics[width =.97\linewidth]{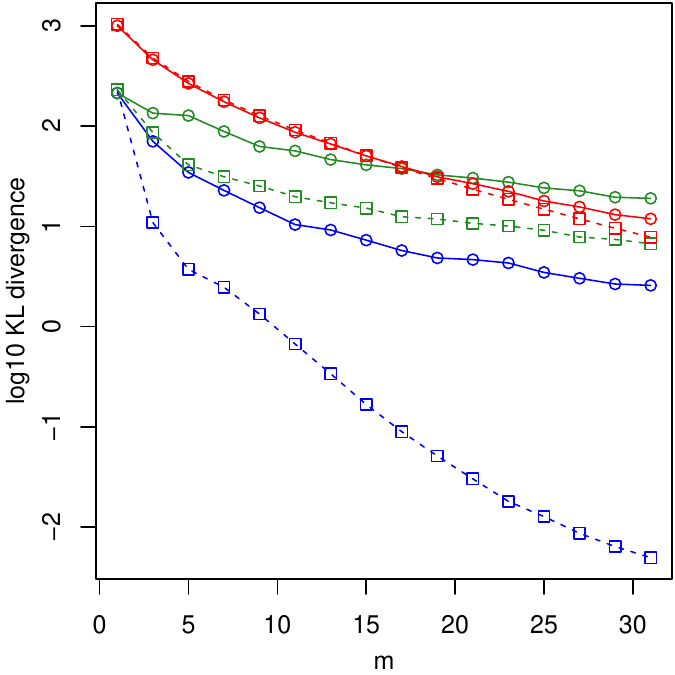}
	\caption{SNR $ =1, \; \nu = 0.5$}
	\label{subfig:kl2dsnr1exp}
	\end{subfigure}%
\hfill
	\begin{subfigure}{.33\textwidth}
	\centering
	\includegraphics[width =.97\linewidth]{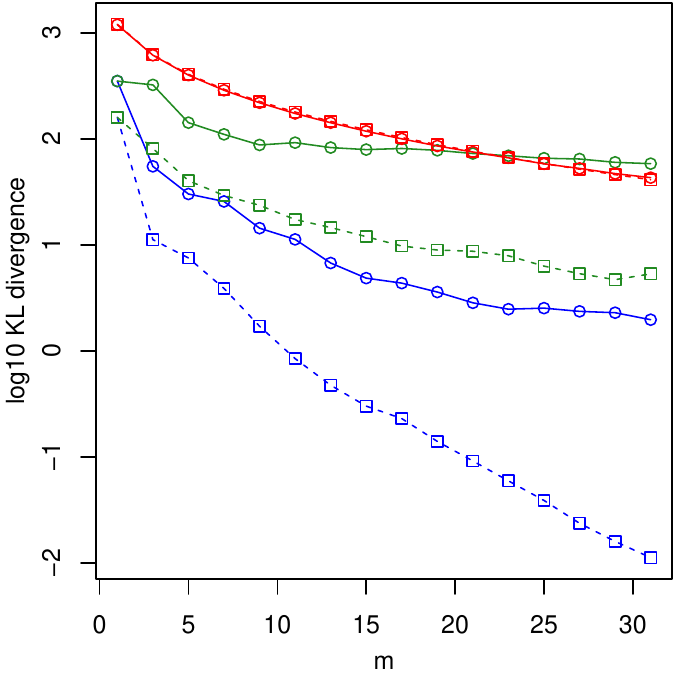}
	\caption{SNR $ =1, \; \nu = 1.5$}
	\label{subfig:kl2dsnr1_15}
	\end{subfigure}%
\hfill
	\begin{subfigure}{.33\textwidth}
	\centering
	\includegraphics[width =.97\linewidth]{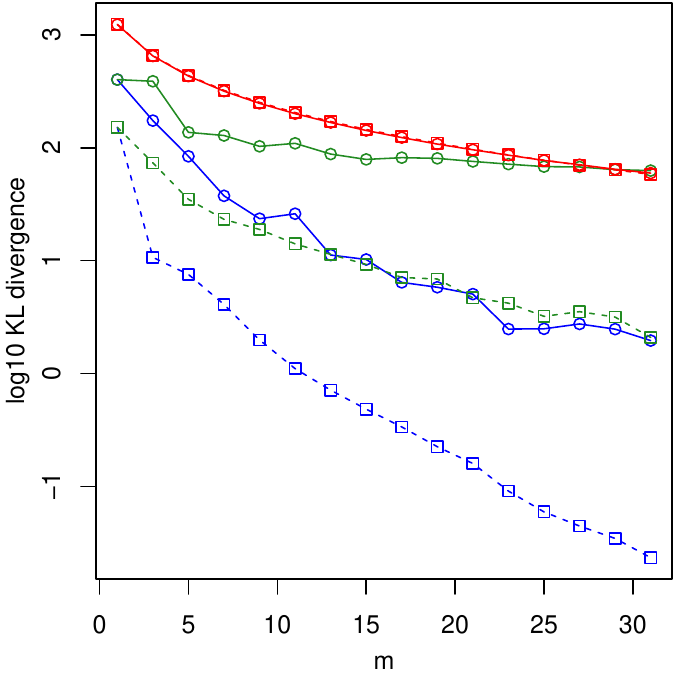}
	\caption{SNR $ =1, \; \nu = 2.5$}
	\end{subfigure}%
\caption{KL divergences (on a log scale) for a Mat\'ern covariance with smoothness $\nu$ on the unit square. Panels (a)--(c): fixed $m=5$, varying signal proportion, with symbols corresponding to (from left to right) SNRs of 0.5, 1, 2, 5, 10, 20, $\infty$, respectively. Panels (d)--(f): fixed SNR $=1$, varying $m$.}
\label{fig:numericalcomp}
\end{figure}

The remaining results are for a two-dimensional domain, $\domain = [0,1]^2$. Exploring Proposition \ref{prop:KLordering}, Figure \ref{fig:numericalcomp} shows KL divergences for different values of $\sigma^2$, $\nu$, and $m$, all for $n_z= $ 6,400 and $\lambda=0.9$. As we can see, the KL divergences for the three methods roughly followed the ordering from Section \ref{sec:sgv}, with latent Vecchia performing better than SGV, which performed better than standard Vecchia. (For SNR=$\infty$, the methods are equivalent.) The screening effect is less clear in two dimensions. Also note that maxmin ordering often resulted in tremendous improvements over coord ordering, except for standard Vecchia, where the two orderings produced similar results.

We also considered very smooth covariances, which are less common in geostatistics but very popular in machine learning. We explored conditioning on the same first $m$ variables in the maxmin ordering, which are spread throughout the domain. Figure \ref{fig:condcomp} shows that this can result in strong improvements over NN ordering for SGV.

\begin{figure}
	\begin{subfigure}{.48\textwidth}
	\centering
	\includegraphics[width =.7\linewidth]{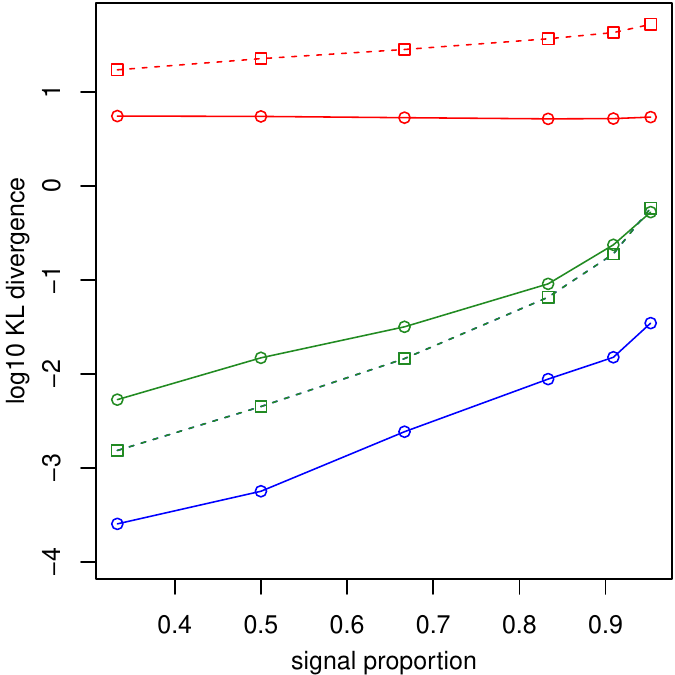}
	\caption{$\nu=3$}
	\end{subfigure}%
\hfill
	\begin{subfigure}{.48\textwidth}
	\centering
	\includegraphics[width =.7\linewidth]{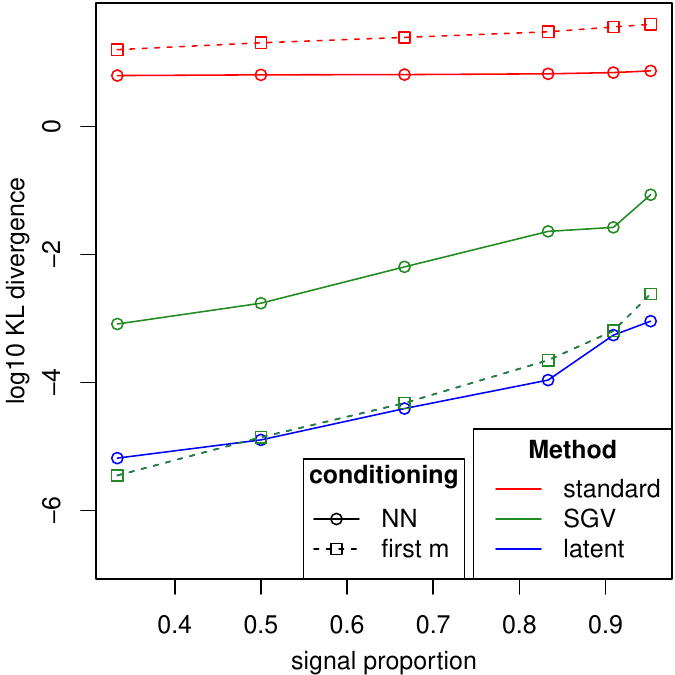}
	\caption{$\nu=10$}
	\end{subfigure}%
\caption{KL divergences (on a log scale) for smooth covariances, comparing nearest-neighbor (NN) conditioning versus always choosing the first $m$ variables; $n_z=400$, $m=16$, maxmin ordering, $\lambda \approx$ 2. For first-$m$ conditioning, SGV and latent are equivalent.}
\label{fig:condcomp}
\end{figure}

The computational feasibility of the methods is explored in Figure \ref{fig:nnzV}, which examines the sparsity of the matrix $\bV$. We can see that SGV keeps the number of nonzero elements per column in $\bV$ at or below $m$, as would be expected from Proposition \ref{prop:complexity}, resulting in linear scaling as a function of $n$. For latent Vecchia, $\bV = \rchol(\bW)$ is considerably denser, and the computational complexity for obtaining $\bV$ scales roughly as $\order(n^2)$, as expected from Proposition \ref{prop:latentcomplexity}. MMD ordering of $\bW$ did not improve the complexity for coord. Figure \ref{fig:Vtime} shows actual computation times for obtaining $\bV$ from $\bW$ using the \texttt{chol} function in the R package \texttt{spam} \citep{furrer2010} on a 4-core machine (Intel Core i7-3770) with 3.4GHz and 16GB RAM. Despite the \texttt{chol} function being more highly optimized for the default MMD ordering than the user-supplied reverse ordering, latent Vecchia with MMD ordering is roughly two orders of magnitude slower than SGV with reverse ordering for $n_z$ around 100{,}000. A further timing study in \citet{Katzfuss2018} shows that, for standard Vecchia and SGV, the time for computing $\bV$ is negligible relative to that for $\bU$; hence, for a given $n$ and $m$, standard Vecchia and SGV require almost the same computation time. In contrast, latent Vecchia can be orders of magnitude slower when $n$ is large.

\begin{figure}
	\begin{subfigure}{.33\textwidth}
	\centering
	\includegraphics[width =.97\linewidth]{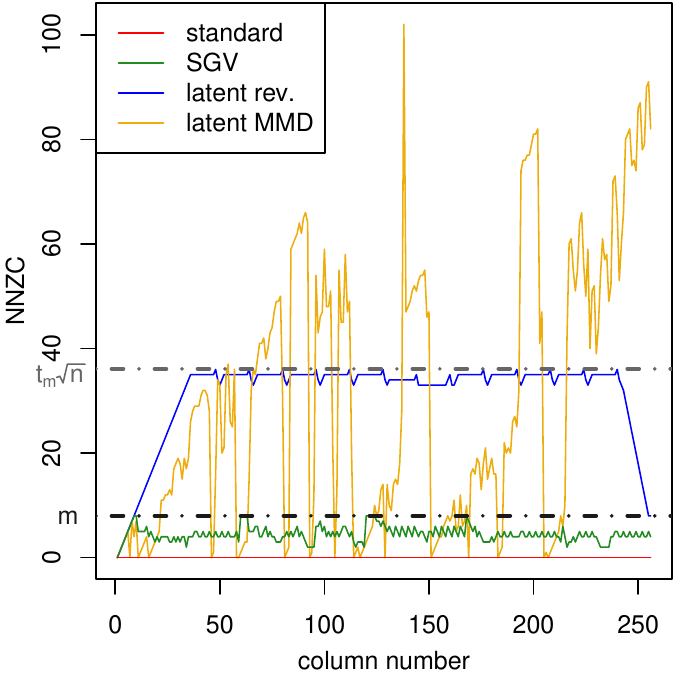}
	\caption{NNZC for $n_z=256$ and coord}
	\end{subfigure}%
\hfill
	\begin{subfigure}{.33\textwidth}
	\centering
	\includegraphics[width =.97\linewidth]{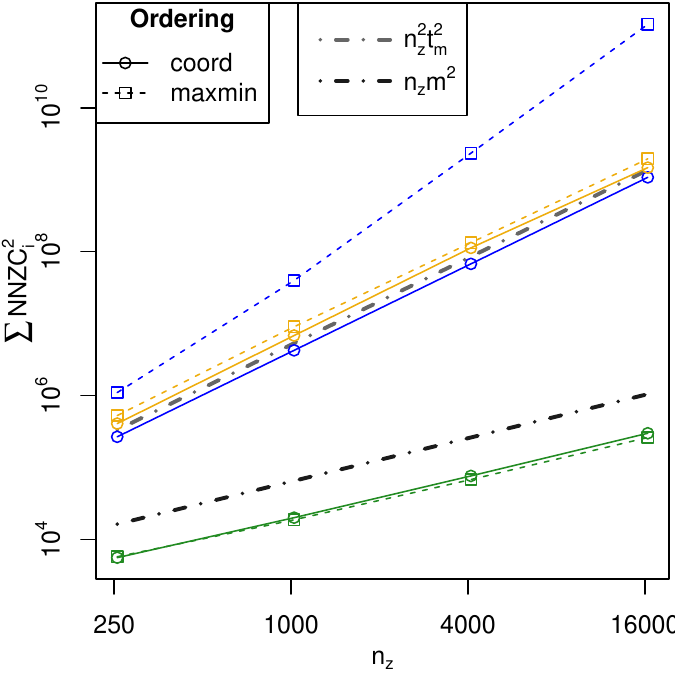}
	\caption{Sum of (NNZC)$^2$}
	\label{fig:complexity}
	\end{subfigure}%
\hfill
	\begin{subfigure}{.33\textwidth}
	\centering
	\includegraphics[width =.97\linewidth]{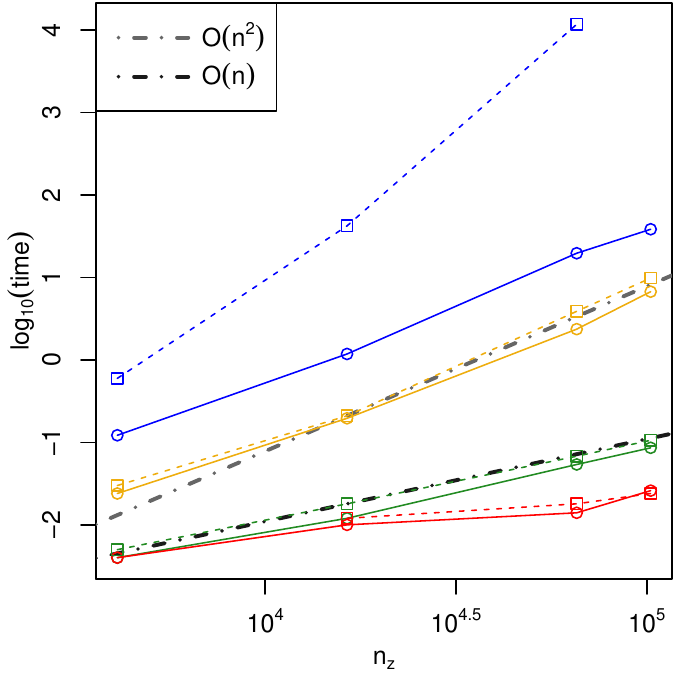}
	\caption{Time to compute $\bV$}
	\label{fig:Vtime}
	\end{subfigure}%
  \caption{Sparsity, complexity, and actual computation times of obtaining $\bV$ with $m=8$, and so $t_m = (2m/\pi)^{1/2} \approx 2.25$ (see proof of Proposition \ref{prop:latentcomplexity}). NNZC: number of nonzero off-diagonal elements per column in $\bV$; MMD and rev.: multiple minimum degree and reverse ordering, respectively, for Cholesky algorithm}
\label{fig:nnzV}
\end{figure}

Figure \ref{fig:mracomp} shows a comparison for large $n$ of four methods that all scale linearly, namely SGV, standard Vecchia, MRA (Section \ref{sec:mra}), and independent blocks (Section \ref{sec:indblocks}), using maxmin ordering where applicable. We set $\domain = [0,1]^2$, $\sigma=\tau=1$, and $\lambda=0.9$. 
In Panel \ref{fig:infill}, we explored the accuracy of the methods under infill asymptotics, by simulating data on a fine $280 \times 280$ grid, and then, starting with a coarse subset or subgrid of size $100 \times 100$, considering larger and larger subsets of the data. It is infeasible to compute the exact KL divergence for large $n_z$, and so we approximated it by subtracting each method's loglikelihood from the loglikelihood for SGV with large $m=40$, all averaged over 10 simulated datasets. While the time complexity for independent blocks and MRA (see Appendix \ref{app:scs}) is only $\order(m^{2/3})$ of that for the Vecchia approaches, SGV outperformed all other approaches even when adjusting for differences in complexity.

\begin{figure}
	\begin{subfigure}{.33\textwidth}
	\centering
	\includegraphics[width =.97\linewidth]{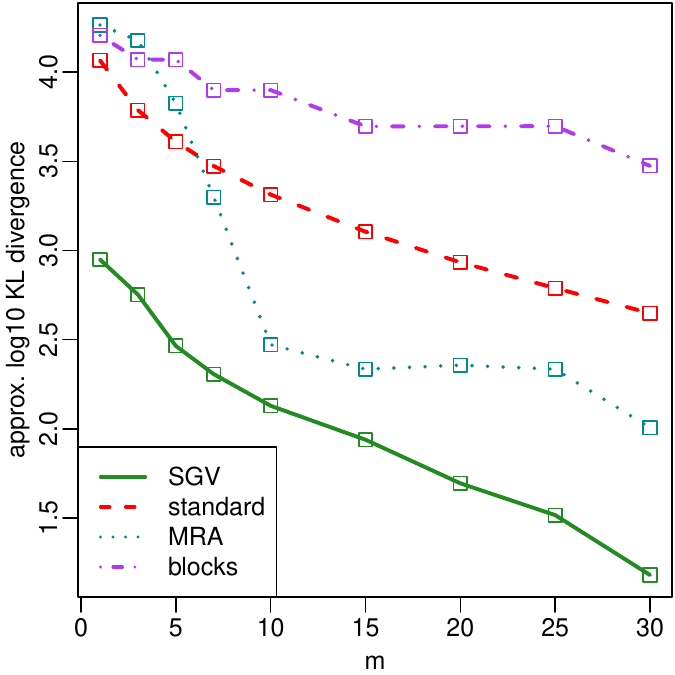}
	\caption{$\nu=0.5$, $n_z=62{,}500$}
	\label{fig:largen05}
	\end{subfigure}%
\hfill
	\begin{subfigure}{.33\textwidth}
	\centering
	\includegraphics[width =.97\linewidth]{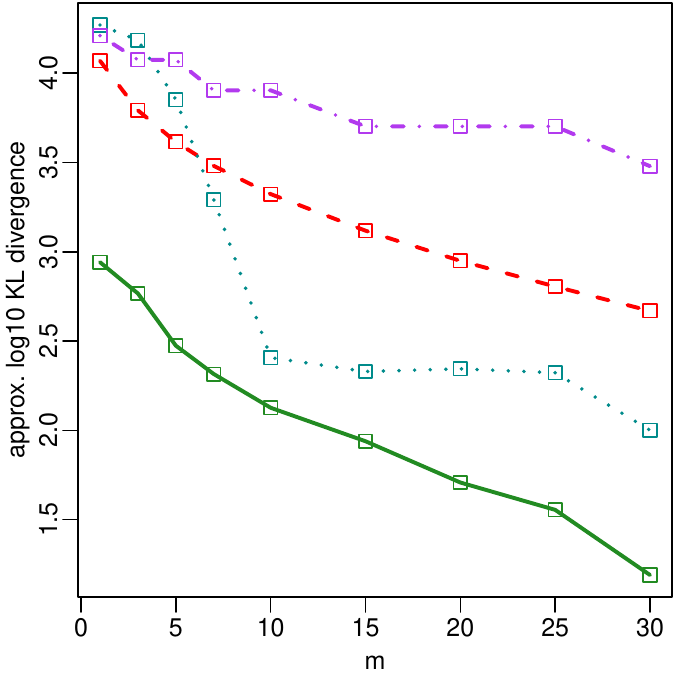}
	\caption{$\nu=1.5$, $n_z=62{,}500$}
	\label{fig:largen15}
	\end{subfigure}%
\hfill
	\begin{subfigure}{.33\textwidth}
	\centering
	\includegraphics[width =.97\linewidth]{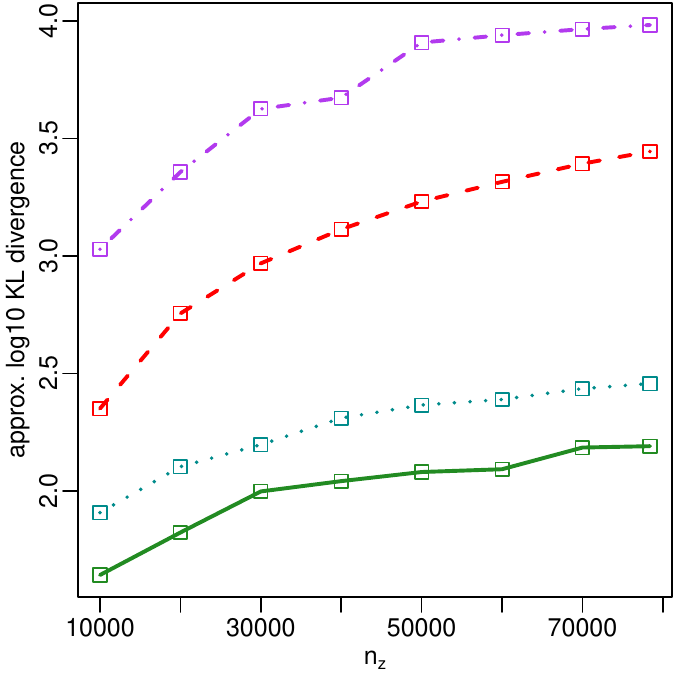}
	\caption{$\nu=1.5$, $m=10$}
	\label{fig:infill}
	\end{subfigure}%
  \caption{Comparison of SGV and standard Vecchia to MRA and independent blocks}
\label{fig:mracomp}
\end{figure}

\section{Conclusions and guidelines \label{sec:conclusions}}

We have presented a general class of sparse GP approximations based on applying Vecchia's approximation to a vector consisting of latent GP realizations and their corresponding noisy observations. Several of the most commonly used GP approximations proposed in the literature are special cases of our class. We provided a formula for fast computation of the likelihood, and we studied the sparsity and computational complexity using connections between Vecchia approaches and directed acyclic graphs. We proposed a novel sparse general Vecchia (SGV), which can dramatically improve upon the approximation accuracy of standard Vecchia while maintaining its linear computational complexity. In contrast, we showed that latent Vecchia (which is used in the nearest-neighbor GP) can scale quadratically in the data size in two-dimensional space.

We now give some guidelines for using the general Vecchia approach in practice. In general, we recommend using our SGV approximation in the presence of nugget or noise, and standard Vecchia if the noise term is zero or almost zero.
In one spatial dimension, left-to-right ordering and nearest-neighbor conditioning is most natural, and SGV is equivalent to latent. In addition, the size $m$ of the conditioning vector can be chosen according to the smoothness (i.e., differentiability at the origin) of the covariance function.
In two-dimensional space, we recommend maxmin ordering. While it is difficult to determine a suitable $m$ a priori, a useful approach is to carry out inference for small $m$, and then gradually increase $m$ until the inference converges or the computational resources are exhausted. Nearest-neighbor conditioning is suitable for low smoothness, while conditioning on the first $m$ latent variables is preferable for higher smoothness when there is a large nugget. This first-$m$ conditioning and its extensions (such as the MRA) has benefits beyond approximation accuracy, such as reduced computational complexity, exact marginal distributions for all variables, and sparse Cholesky factor of the posterior covariance matrix.
While our methods are, in principle, applicable in more than two dimensions, a thorough investigation of their properties in this context is warranted and will be carried out in future work.

The methods and algorithms proposed here are implemented in the R package \texttt{GPvecchia} available at \url{https://github.com/katzfuss-group/GPvecchia}.
\citet{Katzfuss2018} extend the general Vecchia framework to GP prediction at observed and unobserved locations. They also provide further details on computational issues and timing, and an application to a large satellite dataset.

\footnotesize
\appendix
\section*{Acknowledgments}

Katzfuss' research was partially supported by National Science Foundation (NSF) Grant DMS--1521676 and NSF CAREER Grant DMS--1654083. Guinness' research was partially supported by NSF Grant DMS--1613219 and NIH grant R01ES027892. We also acknowledge support from the NSF Research Network for Statistical Methods for Atmospheric and Oceanic Sciences, No.\ 1107046. Wenlong Gong, Marcin Jurek, and Daniel Zilber contributed to the R package \texttt{GPvecchia}.
We would like to thank the reviewers for helpful comments and suggestions.

\section{Vector notation \label{app:vectornotation}}

We define vectors to be objects that contain an ordered list of elements of the same type, equipped with union and intersection operations.
We generally use non-bold lowercase letters for vectors of integers (e.g., $o,p,q$). We use bold lowercase letters (e.g., $\bx,\by,\bz$) for vectors of real numbers or vectors of vectors. Using vectors of vectors makes some of the early definitions slightly cumbersome but greatly simplifies the main unifying results of the paper. Bold uppercase letters usually refer to matrices (e.g., $\bm{C},\bm{K}$), and script letters (e.g., $\mathcal{S}$) for vectors of locations or vectors of vectors of locations.

For example, define $\by = (\by_1,\by_2,\by_3,\by_4,\by_5)$ as a vector of vectors. Subvectoring is accomplished with index vectors and uses subscript notation, for example if $o = (4,1,2)$ is a vector of indices, then $\by_o = (\by_4,\by_1,\by_2)$, respecting the ordering of the index vector.
Unions of vectors are vectors and are defined when the two vectors have the same type and when the ordering of the union is defined. For example, if $\bz = (\bz_1,\bz_2)$, then $\by_o \cup \bz = (\by_4,\by_1, \bz_1, \by_2, \bz_2)$ is a complete definition of the union of $\by_o$ and $\bz$. Likewise, the intersection $\by \cap \bz$ consists of the common elements of the two vectors and is fully defined when the ordering of the intersection is defined.

When the situation demands more abstractness, the ordering of the elements of the union or intersection can be defined via an index function $\#$ that inputs an element and a vector and returns the index occupied by the element in the vector.
Continuing the example above, $\#(\by_4, \by) = 4$, whereas $\#(\by_4, \by_o \cup \bz) = 1$. The index function is vectorized, meaning that $\#( \bz, \by \cup \bz )=(3,5)$ returns the vector of indices occupied by $\bz$ in $\by \cup \bz$. This allows the index function to act as an inverse of the union operator, in the sense that $(\by \cup \bz)_{\#(\bz, \by \cup \bz)} = \bz$.

Vectors whose elements are real numbers are considered as the usual column vectors to which vector addition and multiplication rules apply. Matrices are simply two-dimensional vectors that use double subscripting, and all matrices are viewed as block matrices, with the blocks defined based on context. Functions are vectorized with respect to vectors of locations. For example, if $\locs = (\locs_1,\ldots,\locs_{\ell})$, $\bA = K(\locs,\locs)$ is an $\ell \times \ell$ block matrix with block $\bA_{ij} = K(\locs_i,\locs_j)$. We use $\all$ to represent the vector of all indices, and so $\bA_{i \all} = K(\locs_i,\locs)$.

\section{Review of directed acyclic graphs (DAGs) \label{app:dags}}

Here we provide a brief review of DAGs \citep[see, e.g.,][Sect.~2]{Rutimann2009}. A directed graph consists of vertices, say $\{\bx_1,\ldots,\bx_b\}$, and directed edges (i.e., arrows). Two vertices $\bx_i$ and $\bx_j$ are called adjacent if there is an edge between them. If the edge is directed from $\bx_j$ to $\bx_i$, $\bx_j$ is called a parent of $\bx_i$, and we write $\bx_j \to \bx_i$. If there is no directed edge from $\bx_j$ to $\bx_i$, we write $\bx_j \not\to \bx_i$. A path $\path$ is a sequence of adjacent vertices, and a directed path follows the direction of the arrows. A vertex $\bx_j$ on a path $\path$ is said to be a collider on $\path$ if it has converging arrows on $\path$ (i.e., if the two edges in $\path$ connected to $\bx_j$ both point toward it). If there is a directed path from $\bx_j$ to $\bx_i$, then $\bx_i$ is called a descendant of $\bx_j$. A directed graph is called a DAG if it does not contain directed paths for which the first and last vertices coincide.

For any three disjoint subsets $\mathcal{A},\mathcal{B},\mathcal{C}$ of $\{\bx_1,\ldots,\bx_b\}$, $\mathcal{A}$ and $\mathcal{B}$ are called $d$-separated by $\mathcal{C}$ if, for every (undirected) path $\path$ from a vertex in $\mathcal{A}$ to a vertex in $\mathcal{B}$, there is at least one vertex $\bx_k \in \path$ that blocks the path in one of the following ways:
\begin{description}[itemsep=0.5pt,topsep=1pt,parsep=0.5pt]
\item[B1:] $\bx_k$  is not a collider on $\path$ and $\bx_k$ is in $\mathcal{C}$, or
\item[B2:] $\bx_k$  is a collider on $\path$ and neither $\bx_k$ nor any of its descendants are in $\mathcal{C}$.
\end{description}
If $\bx$ follows a multivariate normal distribution (as we assume here), then $\mathcal{A}$ and $\mathcal{B}$ are conditionally independent given $\mathcal{C}$ if and only if they are $d$-separated by $\mathcal{C}$.

\section{Same conditioning sets (SCS)\label{app:scs}}

In the SCS approach, every $\by_i$ has the same conditioning vector $\by_1$ of size $r_1$; that is, $q(i) = (1)$ for all $i>1$. This is the strategy employed by the MPP and FSA in Sections \ref{sec:mpp}--\ref{sec:fsa} with $r_i=1$ and $r_i = r$, respectively, for $i>1$. For example, one could choose the first $r_1$ variables in the maxmin ordering, which result in a coarse grid over $\domain$.

SCS has several advantages.
First, latent Vecchia automatically adheres to the SGV rules for SCS; or, in other words, the sparsity for the latent approach can be guaranteed.
Second, as discussed in Section \ref{sec:m}, if smoothness and range are large enough, no screening effect will hold. SCS is an extension of the predictive process, which tends to work well in such ``smooth'' situations, because it is equivalent to a Nystr{\"o}m approximation of the leading terms of the Karhunen-Lo\'eve expansion of $y(\cdot)$ \citep{Sang2012}.
Third, a lower computational complexity can be achieved, because $C(\bx_{g(i)},\bx_{g(i)})=C(\by_1,\by_1)$ in \eqref{eq:densproduct} is the same matrix for all $i=2,\ldots,l$ and its Cholesky decomposition only needs to be computed once. Assuming $r_1 = rm$, the cost of the Cholesky decomposition is $\order((rm)^3)$, and each $\bB_i$ and $\bD_i$ can be computed in $\order(r_i^3m^2)$ time, resulting in an overall time complexity for SCS of $\order(nm^2r^2)$ (i.e., reduced by factor $m$ relative to the general case) if $r_i = r$ for $i>1$.
Fourth, the marginal distributions of the $\bx_i$ (and hence also the variances) are exact (see Section \ref{sec:fsa}).
Fifth, $\bV^{-1}$ has the same sparsity structure as $\bV$, which allows fast calculation of the joint posterior predictive distribution for a large number of prediction locations, and extension to Kalman-filter-type inference for massive spatio-temporal data \citep{Jurek2018}.
All of these advantages also hold for the MRA, which can be viewed as an iterative SCS approach at multiple resolutions \citep[][]{Katzfuss2015,Katzfuss2017b,Jurek2018}. However, SCS and MRA may require $r_1 = \order(\sqrt{n_z})$ for accurate approximations in two-dimensional space, which results in a time complexity of $\order(n_z^{3/2})$ \citep{minden2016fast}.

\section{Comparison to composite likelihood\label{sec:clcomp}}

Using simulated data, we compared maximum likelihood estimation (MLE) using our SGV approach to two composite likelihood methods, full-conditional likelihood (FCL) and pairwise-block likelihood (PBL). The data were simulated from a GP model as in Section \ref{sec:gpintro} with exponential covariance, $C(\bs_1,\bs_2) = \sigma^2 \exp( - \| \bs_1 - \bs_2 \| / \alpha )$, where the process and noise variances were known, $\sigma^2=2$ and $\tau^2=1$, respectively, and the task was to estimate the unknown range $\alpha$.

Our first simulation study considered a FCL, defined here as $\adens(\bz) = \prod_{i=1}^n f(z_i | \bz_{-i} )$. As the FCL is expensive to compute, we considered a relatively small grid of size $30 \times 30$ with spacing 1. We simulated 300 datasets with true range $\alpha=10$, and for each dataset $i$ we computed the MLE $\widehat{\alpha}_{ij}$ using each method $j=1,\ldots,5$, namely exact likelihood, FCL, and SGV with $m= 10$, $15$, and $20$. Table \ref{tab:fullcond} contains a summary of the results. We included 95\% confidence intervals for the MSEs, based on a normal approximation of the squared errors. We also computed confidence intervals for the difference in MSEs,
$
( \widehat{\alpha}_{ij} - \alpha )^2 - (\widehat{\alpha}_{iJ} - \alpha )^2,
$
where $J$ corresponds to SGV with $m=20$. FCL was not competitive with SGV.

\begin{table}
\footnotesize \centering
\begin{subtable}{.48\textwidth}
\begin{tabular}{l|rcr}
method & MSE & 95\% CI & CI for diff.\ \\
\hline
exact lik.\ & 3.21 & (2.60, 3.81) & (-0.03, 0.27) \\
FCL & 4.23 & (3.48, 4.97) & (0.61, 1.67) \\
SGV $m\!=\!10$ & 3.22 & (2.63, 3.80) & (-0.15, 0.41) \\
SGV $m\!=\!15$ & 3.07 & (2.50, 3.63) & (-0.20, 0.15) \\
SGV $m\!=\!20$ & 3.09 & (2.53, 3.64) & (0.00, 0.00)
\end{tabular}
\caption{$30 \times 30$ grid}
\label{tab:fullcond}
\end{subtable}
\hfill
\begin{subtable}{.48\textwidth}
\begin{tabular}{l|rrcr}
method & MSE & 95\% CI & CI for diff.\ \\
\hline
PBL 100 bl.\ &   5.71 & (4.79, 6.62) & (0.65, 1.92) \\
PBL 144 bl.\ &   6.26 & (5.22, 7.30) & (1.05, 2.63) \\
PBL 225 bl.\ &   7.05 & (5.73, 8.37) & (1.64, 3.61) \\
SGV $m\!=\!20$ &   4.81 & (4.07, 5.54) & (0.03, 0.74) \\
SGV $m\!=\!40$ &   4.42 & (3.65, 5.19) & (0.00, 0.00)
\end{tabular}
\caption{$100 \times 100$ grid}
\label{tab:pairwise}
\end{subtable}
\caption{Estimation of range parameter from simulated data, including 95\% confidence intervals (CIs) for the difference in MSE relative to the method in the last row}
\label{tab:cl}
\end{table}

Our second simulation study considered PBLs,
$
\prod_{i\sim j} f( \bz_i, \bz_j ),
$
where each $\bz_i$ corresponds to a contiguous rectangle in the spatial domain, and $i \sim j$ means that blocks $i$ and $j$ are spatial neighbors. We used a larger grid of size $100 \times 100$ and a larger range parameter $\alpha = 30$. We again simulated 300 datasets and computed the MLE of $\alpha$ using several settings of the PBL and our SGV. The results are given in Table \ref{tab:pairwise}. 
Note that even SGV with $m=20$ performed better than PBL with 100 blocks of size 100 each.

\section{Illustration of Bayesian inference\label{sec:bayesian}}

This section demonstrates how one can carry out Bayesian inference using the general Vecchia approximation. We used SGV with $m=30$ under
the same settings as in Appendix \ref{sec:clcomp}.

First, we simulated a dataset in the setting of the $30 \times 30$ grid. Based on the prior $\log \alpha \sim \normal(\log(10),0.6^2)$, we evaluated the exact posterior $\dens(\alpha|\bz_o)$ and the posterior $\adens(\alpha|\bz_o)$ implied by the SGV likelihood on a fine grid (see Figure \ref{fig:range_post}). Based on this discrete approximation to the posterior, Figure \ref{fig:pred_post} shows the posterior predictive distribution $\dens(y(\bs_*)|\bz_o)$ at the unobserved point $\bs_* = (15.5,15.5)$ in the center of the grid, along with the distribution $\adens(y(\bs_0)|\bz_o)$ approximated using a general Vecchia prediction method called RF-full in \citet{Katzfuss2018}. The Vecchia posteriors were almost identical to the exact distributions.

We also simulated a dataset in the setting of the $100 \times 100$ grid, randomly selecting $n_z=9{,}000$ data points as observed, with the remaining GP realizations used as test data. Based on the prior $\log \alpha \sim \normal(\log(30),0.6^2)$ and the SGV likelihood, we ran a Metropolis-Hastings sampler for $\log \alpha$ with a normal proposal distribution with standard deviation 0.5 for 1,200 iterations, discarding the first 200 samples and thinning the remaining by a factor of 10. We then computed posterior predictive distributions $\adens(\by_{-o}|\bz_o)$ using RF-full for the 1,000 held-out test locations. Figure \ref{fig:post_intervals} shows the resulting posterior $80\%$ intervals along with the true simulated values of $\by_{-o}$ at the test locations. 79.8\% of the intervals covered the true values, indicating that the posterior predictive distributions obtained using general Vecchia were well calibrated.

\begin{figure}
	\begin{subfigure}{.33\textwidth}
	\centering
	\includegraphics[width =.97\linewidth]{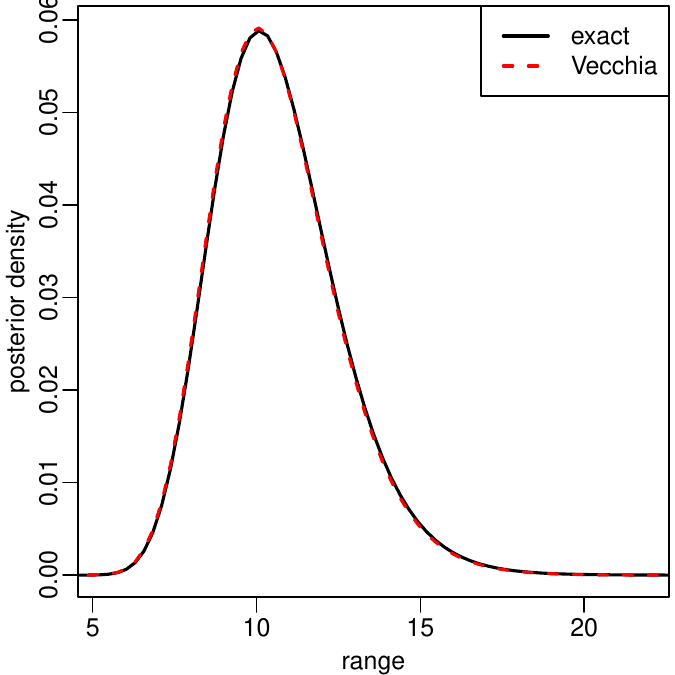}
	\caption{$\dens(\alpha|\bz_o)$ for $30 \times 30$ grid}
	\label{fig:range_post}
	\end{subfigure}%
\hfill
	\begin{subfigure}{.33\textwidth}
	\centering
	\includegraphics[width =.97\linewidth]{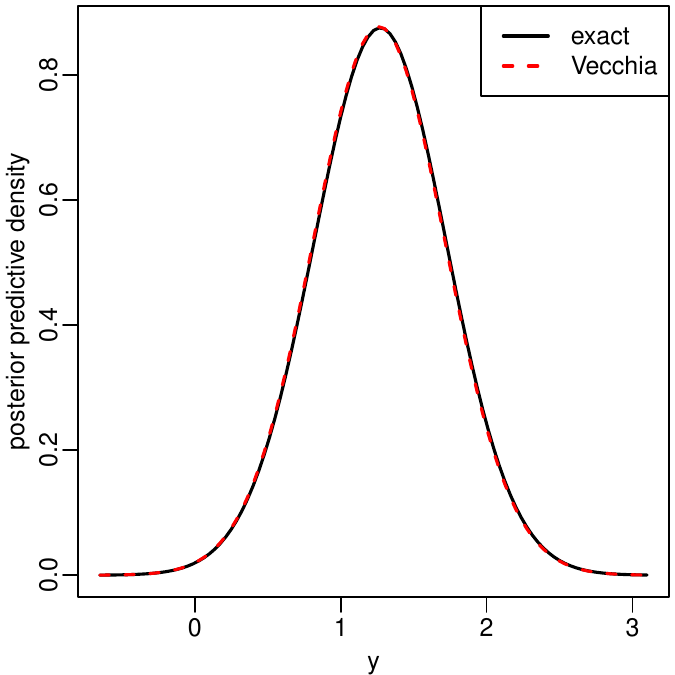}
	\caption{$\dens(y(\bs_*)|\bz_o)$  for $30 \times 30$ grid}
	\label{fig:pred_post}
	\end{subfigure}%
\hfill
	\begin{subfigure}{.33\textwidth}
	\centering
	\includegraphics[width =.97\linewidth]{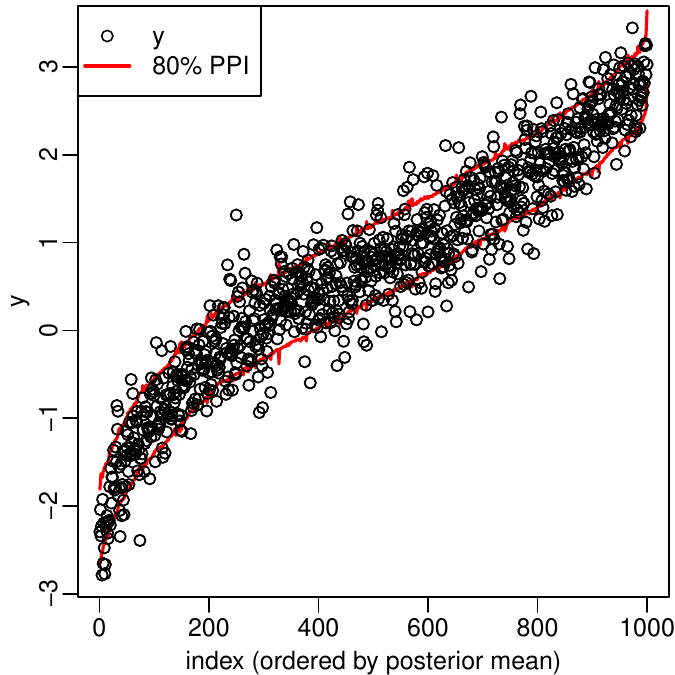}
	\caption{PPIs for $100 \times 100$ grid}
	\label{fig:post_intervals}
	\end{subfigure}%
  \caption{Results for Bayesian inference based on general Vecchia. PPI: posterior predictive interval.}
\label{fig:bayes}
\end{figure}

\section{Proofs \label{app:proofs}}

In this section, we provide proofs for the propositions stated throughout the article.

\begin{proof}[Proof of Proposition \ref{prop:matrixrep}]
For the density in \eqref{eq:densproduct}, we have $\adens(\bx) \propto \exp(-w/2)$, where $w = \sum_{i=1}^b \ba_i'\ba_i$ and
\[
\textstyle\ba_i = (\bD_i^{-1/2})'(\bx_i-\bB_i\bx_{g(i)}) = (\bD_i^{-1/2})'\bx_i + \sum_{j \in g(i)}(-(\bD_i^{-1/2})'\bB_{i}^j\bx_j) = \sum_{j=1}^b \bU_{ji}' \bx_j,
\]
with $\bU$ defined as in \eqref{eq:U}. Hence, $w = \sum_{i=1}^b (\bU_i'\bx)'(\bU_i'\bx) = \bx'\bU\bU'\bx$, where $\bU_i$ is the $i$th block of columns in $\bU$. Because $\bU$ is a nonsingular matrix, we have $\adens(\bx) = \normal_n(\bx | \bfzero, \widehat\bC)$ with $\widehat \bC^{-1} = \bU \bU'$, which proves the first part of the proposition. Note that a proof for a similar expression of the approximate joint density can be found in \citet[][App.~A2]{Datta2016}.

Then, because $\bP$ is a symmetric matrix, we have $\bP = \bP' = \bP^{-1}$ (and $\bP\bM\bP$ results in reverse row-column ordering of the square matrix $\bM$). Thus, we can write
$
\bP \widehat\bC^{-1} \bP = \bP \bU \bU' \bP = \bP \bU \bP\bP \bU' \bP =  (\bP \bU \bP)(\bP \bU \bP)'.
$
The matrix $\bP \bU \bP$ is lower triangular with positive values on the diagonal, and so it must be the Cholesky factor of $\bP \widehat\bC^{-1} \bP'$ since the Cholesky factor is the unique such lower triangular matrix. Therefore, we have $\bU = \bP\chol(\bP\widehat\bC^{-1}\bP)\bP = \rchol(\widehat\bC^{-1})$.
\end{proof}

\begin{proof}[Proof of Proposition \ref{prop:likelihood}]
By rearranging the definition of a conditional density, we obtain
$
\adens(\bz_o) = {\adens(\bx)}/{\adens(\by|\bz_o)},
$
which holds for any $\by$, and so we simply set $\by=\bfzero$. Letting $\bx_0$ be $\bx$ with $\by=\bfzero$, we have
\begin{equation}
\label{eq:likderivation}
\adens(\bz_o) = \frac{\normal(\bx_0|\bfzero,\widehat\bC)}{\normal(\bfzero|\bfmu,\bW^{-1})},
\end{equation}
where $\bfmu \colonequals \E(\by|\bz_o)$ and $\bW \colonequals \var(\by|\bz_o)^{-1}$. For the numerator in \eqref{eq:likderivation}, using the factorization $\widehat\bC^{-1} = \bU \bU'$ supplied by the Vecchia approach, we have $\log |\widehat\bC| = -2 \log |\bU| = \sum_{i=1}^b \log |\bD_i|$ and $\bx_0'\widehat\bC^{-1}\bx_0 = (\bU'\bx_0)'(\bU'\bx_0) = \tilde\bz'\tilde\bz$.

For the denominator in \eqref{eq:likderivation}, according to Theorem 12.2 in \citet{Rue2010}, we have $\bW = \bU_Y\bU_Y'$ and $\bfmu = -\bW^{-1} \bU_Y\bU_Z'\bz_o$. Because $\bW = \bV \bV'$ and $\bV$ is upper triangular, we have $\log |\bW^{-1}|= -2\log |\bV|= - 2\sum_i \log |\bV_{ii}|$. The quadratic form can be obtained as
$
\bfmu'\bW\bfmu = \tilde\bz'\bU_Y'\bW^{-1}\bW\bW^{-1} \bU_Y\tilde\bz = (\bV^{-1}\bU_Y\tilde\bz)'(\bV^{-1}\bU_Y\tilde\bz).
$
\end{proof}

\begin{proof}[Proof of Proposition \ref{prop:sparsity}]
It can be easily verified that $\rchol(\bA) = \bP(\chol(\bP\bA\bP))\bP$ is an upper-triangular matrix for any symmetric, positive-definite $\bA$, because $\chol(\cdot)$ was defined to return the lower-triangular Cholesky factor. Hence, both $\bU = \rchol(\bC^{-1})$ and $\bV=\rchol(\bW)$ are upper triangular. Further, $\bW = \bU_Y \bU_Y'$ is symmetric. Therefore, we only consider the case $j < i$ in the remainder of this proof.

It is well known for precision matrices in multivariate normal distributions \citep[e.g.,][Thm.~12.1]{Rue2010} that $\bW_{ji} = \bfzero$ if $\by_i$ and $\by_j$ are conditionally independent given all other variables in the model (i.e., conditional on $\mathcal{C}_W = \{\by_{-ij},\bz_o\}$).
A similar result for the sparsity of the Cholesky factor \citep[e.g.,][Thm.~12.5]{Rue2010} can be rephrased for our reverse Cholesky decomposition ($\bU = \rchol(\widehat\bC^{-1})$) to say that $\bU_{ji}= \bfzero$ if $\bx_j$ and $\bx_i$ are conditionally independent given $\mathcal{C}_U = \{\bx_{h(i) \setminus j} \}$. For $\bV$, which is the Cholesky factor of the posterior precision matrix $\bW$ (i.e., conditional on $\bz_o$), we have $\bV_{ji} = \bfzero$ if $\by_i$ and $\by_j$ are conditionally independent given $\mathcal{C}_V = \{\by_{h(i) \setminus j},\bz_o\}$. Thus, $\bU_{ji} = \bfzero$ if and only if $\bx_i$ and $\bx_j$ are $d$-separated by $\mathcal{C}_U$ in the DAG, and $\bW_{ji} = \bfzero$ and $\bV_{ji} = \bfzero$ if and only if $\by_i$ and $\by_j$ are $d$-separated by $\mathcal{C}_W$ and $\mathcal{C}_V$, respectively.

Note that $d$-separation cannot hold if $\bx_j \to \bx_i$, and so we only consider (paths between) non-adjacent $\bx_j$ and $\bx_i$ in the remainder of the proof. Any path between such $\bx_j$ and $\bx_i$ must pass through at least one non-collider $\bx_l$ with $l<i$, or through a collider $\by_k$ with $2k - 1 > i$, because arrows in the Vecchia approach can only go forward in the ordering and the only parent for each $\bz_k$ is $\by_k$.
As we have $\mathcal{C}_U = \{\bx_l: l <i , l\neq j \}$, this means that any path between (non-adjacent) $\bx_i$ and $\bx_j$ is either blocked by a non-collider $\bx_l \in \mathcal{C}_U$ (condition B1) or by a collider $\bx_k \notin \mathcal{C}_U$ (B2), which proves part 1.\ of the proposition.

For $\bW$, as all vertices other than $\by_i$ and $\by_j$ are in $\mathcal{C}_W$, we only need to consider condition B1. The only paths between $\by_j$ and $\by_i$ that do not contain a vertex $\bx_k \in \mathcal{C}_W$ that is not a collider (and hence blocks the path), are paths of the form $(\by_j, \by_k, \by_i)$ with $\by_j \to \by_k$ and $\by_i \to \by_k$. This proves part 2.

For $\bV$, any path between $\by_i$ and $\by_j$ that passes through $\mathcal{C}_V = \{\by_{h(i) \setminus j},\bz_o\}$ includes a non-collider in $\mathcal{C}_V$ and is thus blocked (B1). Thus, $\bV_{ji}$ can only be nonzero if there is a path between $\by_i$ and $\by_j$ on the subgraph $\{\by_i,\by_j\} \cup \{\by_k: k>i, \by_k \textnormal{ has at least one observed descendant}\}$.
\end{proof}

\begin{proof}[Proof of Proposition \ref{prop:KLordering}]
Note that for any $p(i) \subset h(i)$, we have $\dens(\by_i|\by_{p(i)}) = \dens(\by_i|\by_{p(i)},\bz_{p(i)})$, due to conditional independence between $\bz_k$ and any other variable in the model given $\by_k$. Thus, for latent Vecchia we can change the conditioning vector of $\by_i$ in \eqref{eq:vecchia2} from $\by_{q(i)}$ to $(\by_{q(i)},\bz_{q(i)})$ without changing the approximation. Likewise, for SGV we can change the conditioning vector of $\by_i$ from $(\by_{q_y(i)},\bz_{q_z(i)})$ to $(\by_{q_y(i)},\bz_{q(i)})$ without changing the approximation, since $q_y(i) \cup q_z(i) = q(i)$. Further, note that $\bz_{q(i)}$ is a subset of $(\by_{q_y(i)},\bz_{q(i)})$, which is in turn a subset of $(\by_{q(i)},\bz_{q(i)})$. Thus, the proposition follows using Thm.~1 in \citet{Guinness2016a}, which says that adding variables to the conditioning vector in Vecchia approximations cannot increase the KL divergence from the true model.
\end{proof}

\begin{proof}[Proof of Proposition \ref{prop:latentcomplexity}]
Without loss of generality, we assume that the locations lie on a regular unit-distance grid on the $d$-dimensional hypercube with $n_z^{1/d}$ unique values in each dimension, and we assume a lexicographic ordering in which locations are ordered first by their first coordinate, for those with same first coordinate by their second coordinate, and so forth.
Let $\bs_i = (s_{i1},\ldots,s_{id})$ be the location of $\by_i$ (and $\bz_i$). Consider a pair of locations $\bs_i$ and $\bs_j$ with $\bs_j = (s_{i1}-t,s_{j2},\ldots,s_{jd})$, which under lexicographic ordering gives $j < i$ when $t>0$. For $\bs_a = (s_{i1}+1,s_{j2},\ldots,s_{jd})$, we have $a>i$. Further, ignoring edge cases, we have $\by_j \to \by_a$ when $1 \leq t \leq t_{m}$, where $t_{m} = \order(m^{1/d})$, since the conditioning vector of $\by_a$ corresponds to the $m$ locations roughly in a semi-ball around $\bs_a$ of radius $t_m$ (e.g., $t_m = (2m/\pi)^{1/2}$ for $d=2$). Also consider $\bs_p = (s_{i1}+1,s_{i2},\ldots,s_{id})$, for which also $p > i$. We can find a path between $\by_a$ and $\by_p$ on the subgraph $\{\by_k : k > i\}$, since all variables on the hyperplane $(s_{i1}+1,\cdot \, ,\ldots,\cdot \, )$ are connected (if $m\geq d$) and have index greater than $i$. We also have $\by_i \to \by_p$. Therefore, there is a path from $\by_j$ to $\by_i$ on the subgraph $\{ \by_j, \by_i \} \cup \{\by_k: k > i\}$, which by Proposition \ref{prop:sparsity} means that $\bV_{ji}$ is non-zero.
Since this is true for any $\bs_j = \{ s_{i1}-t,s_{j2},\ldots,s_{jd} \}$ with $1 \leq t \leq t_{m}$,
there are $\order(n^{1-1/d}m^{1/d})$ nonzero elements in each of the $n_z$ columns of $\bV$, giving a memory complexity of $\order(n^{2-1/d}m^{1/d})$.
As the time complexity for obtaining the reordered Cholesky factor $\bV$ is on the order of the sum of the squares of the number of nonzero elements per column in $\bV$ \citep[e.g.,][Thm.~2.2]{Toledo2007}, the time complexity of obtaining $\bV$ from $\bW$ is $\order( n^{3-2/d} m^{2/d} )$.
\end{proof}

\begin{proof}[Proof of Proposition \ref{prop:complexity}]
First, we show that, for SGV, $\bV$ has at most $mr$ off-diagonal nonzero elements per column. Using Proposition \ref{prop:sparsity}, that means that we need to show that, for any $j<i$, there is no path between $\by_i$ and $\by_j$ on the subgraph $\mathcal{G}_{ij}^\ell$ if $\by_j \not\to \by_i$, where $\mathcal{G}_{ij}^k \colonequals \{\by_i,\by_j\} \cup \{\by_t: \max(i,j)<t \leq k\}$. (Note that this statement then also holds if we restrict the subgraph to vertices with observed descendants.) Define $D_i^k \colonequals \{t: \by_t \textnormal{ is a descendant of } \by_i \textnormal{ in } \mathcal{G}_{ij}^k\}$, and analogously for $D_j^k$. Thus, assuming that $\by_j \not\to \by_i$, we need to show that $D_i^\ell \cap D_j^\ell = \emptyset$, which we will do by induction. We have $\mathcal{G}_{ij}^{i+1} = \{\by_i, \by_j, \by_k\}$, where $q_y(k)$ can only contain either $i$ or $j$ by the rules of the SGV, because $j \notin q_y(i)$, and so $D_i^{i+1} \cap D_j^{i+1} = \emptyset$. Now, assume that $D_i^{k} \cap D_j^{k} = \emptyset$ for $k>i$. Then, for any $t_i \in D_i^k$ and $t_j \in D_j^k$, $\by_{t_i}$ and $\by_{t_j}$ cannot be adjacent. Hence, by the rules of the SGV, $q_y(k+1)$ can only contain either elements of $D_i^k$ or of $D_j^k$, and so $D_i^{k+1} \cap D_j^{k+1} = \emptyset$. In summary, for the SGV and $j<i$, $\bV_{ji} = \bfzero$ unless $\by_j \to \by_i$, and so $\bV$ has at most $mr$ off-diagonal elements per column.

The time complexity for obtaining the reordered Cholesky factor $\bV$ (and the selected inverse of $\bW$) is on the order of the sum of the squares of the number of nonzero elements per column in $\bV$ \citep[e.g.,][Thm.~2.2]{Toledo2007}. Hence, the time complexity for computing $\bW$, its decomposition, and its selected inverse is $\order(nm^2r^2)$. The time and memory complexity for computing $\bU$ is $\order(nm^3r^2)$ and $\order(nmr)$ (i.e., at least as high as that for computing $\bV$), and so SGV has the same computational complexity as standard Vecchia.
\end{proof}

\footnotesize
\bibliographystyle{apalike}
\bibliography{vecchiabib}

\begin{thebibliography}{}

\bibitem[Banerjee et~al., 2004]{Banerjee2004}
Banerjee, S., Carlin, B.~P., and Gelfand, A.~E. (2004).
\newblock {\em {Hierarchical Modeling and Analysis for Spatial Data}}.
\newblock Chapman {\&} Hall.

\bibitem[Banerjee et~al., 2008]{Banerjee2008}
Banerjee, S., Gelfand, A.~E., Finley, A.~O., and Sang, H. (2008).
\newblock {Gaussian predictive process models for large spatial data sets}.
\newblock {\em Journal of the Royal Statistical Society, Series B},
  70(4):825--848.

\bibitem[Cressie and Davidson, 1998]{cressie1998image}
Cressie, N. and Davidson, J.~L. (1998).
\newblock {Image analysis with partially ordered Markov models}.
\newblock {\em {Computational Statistics \& Data Analysis}}, 29(1):1--26.

\bibitem[Cressie and Johannesson, 2008]{Cressie2008}
Cressie, N. and Johannesson, G. (2008).
\newblock {Fixed rank kriging for very large spatial data sets}.
\newblock {\em Journal of the Royal Statistical Society, Series B},
  70(1):209--226.

\bibitem[Cressie and Wikle, 2011]{Cressie2011}
Cressie, N. and Wikle, C.~K. (2011).
\newblock {\em {Statistics for Spatio-Temporal Data}}.
\newblock Wiley, Hoboken, NJ.

\bibitem[Datta et~al., 2016a]{Datta2016}
Datta, A., Banerjee, S., Finley, A.~O., and Gelfand, A.~E. (2016a).
\newblock {Hierarchical nearest-neighbor Gaussian process models for large
  geostatistical datasets}.
\newblock {\em Journal of the American Statistical Association},
  111(514):800--812.

\bibitem[Datta et~al., 2016b]{Datta2016b}
Datta, A., Banerjee, S., Finley, A.~O., and Gelfand, A.~E. (2016b).
\newblock {On nearest-neighbor Gaussian process models for massive spatial
  data}.
\newblock {\em Wiley Interdisciplinary Reviews: Computational Statistics},
  8(5):162--171.

\bibitem[Datta et~al., 2016c]{Datta2016a}
Datta, A., Banerjee, S., Finley, A.~O., Hamm, N. A.~S., and Schaap, M. (2016c).
\newblock {Non-separable dynamic nearest-neighbor Gaussian process models for
  large spatio-temporal data with an application to particulate matter
  analysis}.
\newblock {\em Annals of Applied Statistics}, 10(3):1286--1316.

\bibitem[Du et~al., 2009]{Du2009}
Du, J., Zhang, H., and Mandrekar, V.~S. (2009).
\newblock {Fixed-domain asymptotic properties of tapered maximum likelihood
  estimators}.
\newblock {\em The Annals of Statistics}, 37:3330--3361.

\bibitem[Eidsvik et~al., 2014]{Eidsvik2012}
Eidsvik, J., Shaby, B.~A., Reich, B.~J., Wheeler, M., and Niemi, J. (2014).
\newblock {Estimation and prediction in spatial models with block composite
  likelihoods using parallel computing}.
\newblock {\em Journal of Computational and Graphical Statistics},
  23(2):295--315.

\bibitem[Eubank and Wang, 2002]{Eubank2002}
Eubank, R.~L. and Wang, S.~J. (2002).
\newblock {The equivalence between the Cholesky decomposition and the Kalman
  filter}.
\newblock {\em The American Statistician}, 56(1):39--43.

\bibitem[Finley et~al., 2017]{Finley2017}
Finley, A.~O., Datta, A., Cook, B.~C., Morton, D.~C., Andersen, H.~E., and
  Banerjee, S. (2017).
\newblock {Applying nearest neighbor Gaussian processes to massive spatial data
  sets: Forest canopy height prediction across Tanana Valley Alaska}.
\newblock {\em arXiv:1702.00434}.

\bibitem[Finley et~al., 2009]{Finley2009}
Finley, A.~O., Sang, H., Banerjee, S., and Gelfand, A.~E. (2009).
\newblock {Improving the performance of predictive process modeling for large
  datasets}.
\newblock {\em Computational Statistics {\&} Data Analysis}, 53(8):2873--2884.

\bibitem[Furrer et~al., 2006]{furrer2006covariance}
Furrer, R., Genton, M.~G., and Nychka, D. (2006).
\newblock Covariance tapering for interpolation of large spatial datasets.
\newblock {\em Journal of Computational and Graphical Statistics},
  15(3):502--523.

\bibitem[Furrer and Sain, 2010]{furrer2010}
Furrer, R. and Sain, S.~R. (2010).
\newblock {spam}: A sparse matrix {R} package with emphasis on {MCMC} methods
  for {G}aussian {M}arkov random fields.
\newblock {\em Journal of Statistical Software}, 36(10):1--25.

\bibitem[Gerber et~al., 2018]{Gerber2018}
Gerber, F., Furrer, R., Schaepman-Strub, G., de~Jong, R., and Schaepman, M.~E.
  (2018).
\newblock {Predicting missing values in spatio-temporal satellite data}.
\newblock {\em IEEE Transactions on Geoscience and Remote Sensing},
  56:2841--2853.

\bibitem[Gramacy and Apley, 2015]{Gramacy2015}
Gramacy, R.~B. and Apley, D.~W. (2015).
\newblock {Local Gaussian process approximation for large computer
  experiments}.
\newblock {\em Journal of Computational and Graphical Statistics},
  24(2):561--578.

\bibitem[Gramacy and Lee, 2012]{gramacy2012cases}
Gramacy, R.~B. and Lee, H.~K. (2012).
\newblock Cases for the nugget in modeling computer experiments.
\newblock {\em Statistics and Computing}, 22(3):713--722.

\bibitem[Guhaniyogi and Banerjee, 2018]{Guhaniyogi2018}
Guhaniyogi, R. and Banerjee, S. (2018).
\newblock Meta-kriging: Scalable bayesian modeling and inference for massive
  spatial datasets.
\newblock {\em Technometrics}, 60(4):430--444.

\bibitem[Guinness, 2018]{Guinness2016a}
Guinness, J. (2018).
\newblock {Permutation methods for sharpening Gaussian process approximations}.
\newblock {\em Technometrics}, 60(4):415--429.

\bibitem[Heaton et~al., 2019]{Heaton2017}
Heaton, M.~J., Datta, A., Finley, A.~O., Furrer, R., Guinness, J., Guhaniyogi,
  R., Gerber, F., Gramacy, R.~B., Hammerling, D., Katzfuss, M., Lindgren, F.,
  Nychka, D.~W., Sun, F., and Zammit-Mangion, A. (2019).
\newblock {A case study competition among methods for analyzing large spatial
  data}.
\newblock {\em Journal of Agricultural, Biological, and Environmental
  Statistics}, accepted.

\bibitem[Higdon, 1998]{Higdon1998}
Higdon, D. (1998).
\newblock {A process-convolution approach to modelling temperatures in the
  North Atlantic Ocean}.
\newblock {\em Environmental and Ecological Statistics}, 5(2):173--190.

\bibitem[Huang and Sun, 2018]{Huang2016}
Huang, H. and Sun, Y. (2018).
\newblock {Hierarchical low rank approximation of likelihoods for large spatial
  datasets}.
\newblock {\em Journal of Computational and Graphical Statistics},
  27(1):110--118.

\bibitem[Jurek and Katzfuss, 2018]{Jurek2018}
Jurek, M. and Katzfuss, M. (2018).
\newblock {Multi-resolution filters for massive spatio-temporal data}.
\newblock {\em arXiv:1810.04200}.

\bibitem[Kalman, 1960]{Kalman1960}
Kalman, R. (1960).
\newblock {A new approach to linear filtering and prediction problems}.
\newblock {\em Journal of Basic Engineering}, 82(1):35--45.

\bibitem[Katzfuss, 2017]{Katzfuss2015}
Katzfuss, M. (2017).
\newblock {A multi-resolution approximation for massive spatial datasets}.
\newblock {\em Journal of the American Statistical Association},
  112(517):201--214.

\bibitem[Katzfuss and Cressie, 2011]{Katzfuss2010}
Katzfuss, M. and Cressie, N. (2011).
\newblock {Spatio-temporal smoothing and EM estimation for massive
  remote-sensing data sets}.
\newblock {\em Journal of Time Series Analysis}, 32(4):430--446.

\bibitem[Katzfuss and Gong, 2019]{Katzfuss2017b}
Katzfuss, M. and Gong, W. (2019).
\newblock {A class of multi-resolution approximations for large spatial
  datasets}.
\newblock {\em Statistica Sinica}, accepted.

\bibitem[Katzfuss et~al., 2018]{Katzfuss2018}
Katzfuss, M., Guinness, J., Gong, W., and Zilber, D. (2018).
\newblock {Vecchia approximations of Gaussian-process predictions}.
\newblock {\em arXiv:1805.03309}.

\bibitem[Kaufman et~al., 2008]{kaufman2008covariance}
Kaufman, C.~G., Schervish, M.~J., and Nychka, D.~W. (2008).
\newblock Covariance tapering for likelihood-based estimation in large spatial
  data sets.
\newblock {\em Journal of the American Statistical Association},
  103(484):1545--1555.

\bibitem[Lauritzen, 1996]{Lauritzen1996}
Lauritzen, S.~L. (1996).
\newblock {\em Graphical Models}.
\newblock Clarendon Press.

\bibitem[Lindgren et~al., 2011]{lindgren2011explicit}
Lindgren, F., Rue, H., and Lindstr{\"o}m, J. (2011).
\newblock {An explicit link between Gaussian fields and Gaussian Markov random
  fields: the stochastic partial differential equation approach}.
\newblock {\em Journal of the Royal Statistical Society: Series B},
  73(4):423--498.

\bibitem[Minden et~al., 2016]{minden2016fast}
Minden, V., Damle, A., Ho, K.~L., and Ying, L. (2016).
\newblock {Fast spatial Gaussian process maximum likelihood estimation via
  skeletonization factorizations}.
\newblock {\em arXiv preprint arXiv:1603.08057}.

\bibitem[Nychka et~al., 2015]{Nychka2012}
Nychka, D.~W., Bandyopadhyay, S., Hammerling, D., Lindgren, F., and Sain, S.~R.
  (2015).
\newblock {A multi-resolution Gaussian process model for the analysis of large
  spatial data sets}.
\newblock {\em Journal of Computational and Graphical Statistics},
  24(2):579--599.

\bibitem[Qui{\~{n}}onero-Candela and Rasmussen, 2005]{Quinonero-Candela2005}
Qui{\~{n}}onero-Candela, J. and Rasmussen, C.~E. (2005).
\newblock {A unifying view of sparse approximate Gaussian process regression}.
\newblock {\em Journal of Machine Learning Research}, 6:1939--1959.

\bibitem[Rasmussen and Williams, 2006]{Rasmussen2006}
Rasmussen, C.~E. and Williams, C. K.~I. (2006).
\newblock {\em {Gaussian Processes for Machine Learning}}.
\newblock MIT Press.

\bibitem[Rauch et~al., 1965]{Rauch1965}
Rauch, H., Rauch, H., Tung, F., and Striebel, C. (1965).
\newblock {Maximum likelihood estimates of linear dynamic systems}.
\newblock {\em AIAA Journal}, 3(8):1445--1450.

\bibitem[Rue and Held, 2005]{rue2005gaussian}
Rue, H. and Held, L. (2005).
\newblock {\em {Gaussian Markov Random Fields: Theory and Applications}}.
\newblock CRC press.

\bibitem[Rue and Held, 2010]{Rue2010}
Rue, H. and Held, L. (2010).
\newblock {Discrete spatial variation}.
\newblock In {\em Handbook of Spatial Statistics}, chapter~12, pages 171--200.
  CRC Press.

\bibitem[R{\"{u}}timann and B{\"{u}}hlmann, 2009]{Rutimann2009}
R{\"{u}}timann, P. and B{\"{u}}hlmann, P. (2009).
\newblock {High dimensional sparse covariance estimation via directed acyclic
  graphs}.
\newblock {\em Electronic Journal of Statistics}, 3:1133--1160.

\bibitem[Sang and Huang, 2012]{Sang2012}
Sang, H. and Huang, J.~Z. (2012).
\newblock {A full scale approximation of covariance functions}.
\newblock {\em Journal of the Royal Statistical Society, Series B},
  74(1):111--132.

\bibitem[Sang et~al., 2011]{Sang2011a}
Sang, H., Jun, M., and Huang, J.~Z. (2011).
\newblock {Covariance approximation for large multivariate spatial datasets
  with an application to multiple climate model errors}.
\newblock {\em Annals of Applied Statistics}, 5(4):2519--2548.

\bibitem[Shaby, 2014]{Shaby2014}
Shaby, B.~A. (2014).
\newblock {The open-faced sandwich adjustment for MCMC using estimating
  functions}.
\newblock {\em Journal of Computational and Graphical Statistics},
  23(3):853--876.

\bibitem[Snelson and Ghahramani, 2007]{Snelson2007}
Snelson, E. and Ghahramani, Z. (2007).
\newblock {Local and global sparse Gaussian process approximations}.
\newblock In {\em Artificial Intelligence and Statistics 11 (AISTATS)}.

\bibitem[Stein, 2011]{Stein2011}
Stein, M.~L. (2011).
\newblock {When does the screening effect hold?}
\newblock {\em The Annals of Statistics}, 39(6):2795--2819.

\bibitem[Stein, 2014]{Stein2013a}
Stein, M.~L. (2014).
\newblock {Limitations on low rank approximations for covariance matrices of
  spatial data}.
\newblock {\em Spatial Statistics}, 8:1--19.

\bibitem[Stein et~al., 2004]{stein2004}
Stein, M.~L., Chi, Z., and Welty, L. (2004).
\newblock {Approximating likelihoods for large spatial data sets}.
\newblock {\em Journal of the Royal Statistical Society: Series B},
  66(2):275--296.

\bibitem[Sun and Stein, 2016]{Sun2016}
Sun, Y. and Stein, M.~L. (2016).
\newblock {Statistically and computationally efficient estimating equations for
  large spatial datasets}.
\newblock {\em Journal of Computational and Graphical Statistics},
  25(1):187--208.

\bibitem[Toledo, 2007]{Toledo2007}
Toledo, S. (2007).
\newblock {Lecture Notes on Combinatorial Preconditioners, Chapter 3}.
\newblock http://www.tau.ac.il/{\~{}}stoledo/Support/chapter-direct.pdf.

\bibitem[Varin et~al., 2011]{Varin2011}
Varin, C., Reid, N., and Firth, D. (2011).
\newblock An overview of composite likelihood methods.
\newblock {\em Statistica Sinica}, pages 5--42.

\bibitem[Vecchia, 1988]{Vecchia1988}
Vecchia, A. (1988).
\newblock {Estimation and model identification for continuous spatial
  processes}.
\newblock {\em Journal of the Royal Statistical Society, Series B},
  50(2):297--312.

\bibitem[Wikle and Cressie, 1999]{Wikle1999}
Wikle, C.~K. and Cressie, N. (1999).
\newblock {A dimension-reduced approach to space-time Kalman filtering}.
\newblock {\em Biometrika}, 86(4):815--829.

\bibitem[Zhang et~al., 2018]{Zhang2018}
Zhang, B., Sang, H., and Huang, J.~Z. (2018).
\newblock {Smoothed full-scale approximation of Gaussian process models for
  computation of large spatial datasets}.
\newblock {\em Statistica Sinica}, accepted.

\bibitem[Zhang, 2012]{Zhang2012}
Zhang, H. (2012).
\newblock {Asymptotics and computation for spatial statistics}.
\newblock In {\em Advances and Challenges in Space-time Modelling of Natural
  Events}, pages 239--252. Springer.

\end{thebibliography}

\end{document}